\documentclass[12pt,british,english]{article}
\usepackage{mathpazo}
\usepackage[T1]{fontenc}
\usepackage[latin9]{inputenc}
\usepackage{geometry}
\usepackage{babel}
\usepackage[dvipsnames]{xcolor}
\usepackage{amsmath}
\usepackage{amsthm}
\usepackage{amssymb}
\usepackage{xfrac}
\usepackage{graphicx}
\usepackage{setspace}
\usepackage{esint}
\usepackage[title]{appendix}
\usepackage[authoryear]{natbib}
\PassOptionsToPackage{normalem}{ulem}
\usepackage{ulem}
\usepackage[colorlinks=true,allcolors=Blue]{hyperref}

\makeatletter
  \theoremstyle{plain}
  \newtheorem{assumption}{\protect\assumptionname}
  \theoremstyle{plain}
  \newtheorem{lemma}{\protect\lemmaname}
  \theoremstyle{plain}
  \newtheorem{proposition}{\protect\propositionname}
  \theoremstyle{plain}
  \newtheorem{corollary}{\protect\corollaryname}
  \theoremstyle{plain}
  \newtheorem{claim}{\protect\claimname}

\usepackage{babel}
\usepackage[bottom]{footmisc}
\usepackage{bbm}
\usepackage{multicol}

\usepackage{booktabs}
\usepackage{multirow}

\newcommand{\tablenote}{\vspace{1mm} \footnotesize \setstretch{1.05}}

\@ifundefined{showcaptionsetup}{}{%
 \PassOptionsToPackage{caption=false}{subfig}}
\usepackage{subfig}
\makeatother

  \addto\captionsbritish{\renewcommand{\assumptionname}{Assumption}}
  \addto\captionsbritish{\renewcommand{\lemmaname}{Lemma}}
  \addto\captionsbritish{\renewcommand{\propositionname}{Proposition}}
  \addto\captionsenglish{\renewcommand{\assumptionname}{Assumption}}
  \addto\captionsenglish{\renewcommand{\lemmaname}{Lemma}}
  \addto\captionsenglish{\renewcommand{\propositionname}{Proposition}}
  \providecommand{\assumptionname}{Assumption}
  \providecommand{\lemmaname}{Lemma}
  \providecommand{\propositionname}{Proposition}
\addto\captionsbritish{\renewcommand{\corollaryname}{Corollary}}
\addto\captionsenglish{\renewcommand{\corollaryname}{Corollary}}
\providecommand{\corollaryname}{Corollary}
\addto\captionsbritish{\renewcommand{\claimname}{Claim}}
\addto\captionsenglish{\renewcommand{\claimname}{Claim}}
\providecommand{\claimname}{Claim}

\begin{document}
\begin{titlepage}
\title{Relational Communication\thanks{
We thank Luis Zerme{\~n}o for important early contributions to this paper. We also thank anonymous referees, Ricardo Alonso, David Austen-Smith, Simon Board, Wouter Dessein, Sven Feldmann, Yuk-fai Fong, Robert Gibbons, Richard Holden, Johannes H{\"o}rner, Navin Kartik, Jin Li, Niko Matouschek, Tymofiy Mylovanov, Marco Ottaviani, Michael Powell, Heikki Rantakari, Joel Sobel, Kathryn Spier, and participants at various seminars and conferences for helpful comments and suggestions. Kolotilin acknowledges support from the Australian Research Council Discovery Early Career Research Award DE160100964 and from MIT Sloan's Program on Innovation in Markets and Organizations.}}
\author{Anton Kolotilin and Hongyi Li\thanks{%
UNSW Business School, School of Economics. Emails:
\href{mailto:akolotilin@gmail.com}{akolotilin@gmail.com} and \href{mailto:hongyi@hongyi.li}{hongyi@hongyi.li}.}}
\date{10th December 2020}
\clearpage\maketitle
\thispagestyle{empty}

\begin{abstract}
We study a communication game between an informed sender and an uninformed receiver with repeated interactions and voluntary transfers. Transfers motivate the receiver's decision-making and signal the sender's information. Although full separation can always be supported in equilibrium, partial or complete pooling is optimal if the receiver's decision-making is highly responsive to information. In this case, the receiver's decision-making is disciplined by pooling states where she is most tempted to defect.

\bigskip

\noindent \textbf{JEL\ Classification:}\ C73, D82, D83

\end{abstract}

\end{titlepage}

\setcounter{page}{1}
\section{Introduction\label{Intro}}

Decision-makers and informed parties often develop relationships in which communication and decision-making are governed by informal agreements. We study how such interactions can be disciplined using relational contracts: discretionary compensation schemes that are self-enforcing in a repeated game. We characterize communication and decision-making patterns in optimal equilibria.

As an example of such relational communication, consider the relationships between political advocates and politicians. Political advocates seek to influence politicians, often over a broad range of policy decisions; think, for instance, of the Koch Brothers' extensive lobbying activities. They ply politicians with information about the electoral consequences of various policy choices, such as voter attitudes toward clean-energy legislation or gun control. They also make transfers to politicians, in the form of political contributions. Such transfers serve as contingent payments for favorable policy decisions (\citealt{grossman1994protection,Grossman:1996ck}) and credible signals of advocates' information (\citealt{AustenSmith:1995cd} and \citealt{Lohmann:1995jf}).  Because pay-to-play in this context is illegal, agreements between politicians and advocates are largely implicit and sustained within long-standing relationships.

Another example of relational communication is in repeated principal-agent relationships where an agent implements a series of projects and a principal has relevant information about these projects. Think, for instance, of a movie director and a studio, such as Akira Kurosawa and Toho Studios. The director has to decide whether to position each movie project as more mainstream (in which case the movie will be a sure-fire box-office draw) or as more art-house (in which case profitability is uncertain, but potentially large). The studio may provide informed advice or recommendations about potential box-office revenues, but the director retains creative control. The studio may also make payments to the director -- in the form of bonuses, perks, or additional funding -- to reward the director's decision-making, or to bolster the credibility of the studio head's advice.\footnote{Relatedly, \cite{Hermalin:1998tw} and \cite{benabou2003intrinsic} discuss how payments from principal to agent can be used not only to reward performance, but also to credibly signal the principal's private information.}

To study relational communication, we add repeated interactions and voluntary transfers, as in \cite{levin2003relational}, to the \cite{crawford1982strategic} model of strategic communication. In each period, the sender privately observes an independent draw of the state and sends a message to the receiver, who then makes a decision. The players' preferred decisions are increasing in the state, but the magnitude and sign of the difference between preferred decisions may depend on the state. The players can make voluntary transfers to each other at any point in the game.

In relational communication, transfers allow the sender not only to reward the receiver for compliant decision-making, but also to credibly signal his private information. In particular, full separation can be supported in equilibrium, even when the players are impatient. Therefore, the essential incentive constraint is that the receiver is tempted to make decisions that benefit herself but hurt the sender.

We show that a message rule can be supported in equilibrium if and only if it is \emph{monotone}: it induces a monotone partition of the set of states. In any (Pareto) optimal equilibrium, the decision rule simply maximizes, subject to the receiver's incentive constraint, the joint payoff for each message. Therefore, given this decision rule, the optimal message rule solves the monotone persuasion problem: it maximizes the expected joint payoff over all monotone message rules. This is the Bayesian persuasion problem (\citealt{Rayo:2010jr} and \citealt{kamenica2011bayesian}) with an additional constraint that message rules must be monotone. 

We completely characterize the optimal (second-best) equilibrium when the players' payoffs are quadratic. 
Our key insights are about how information should be selectively hidden and revealed to manage decision-making within an optimal relationship. 
Consider \emph{high-conflict} states where conflict of interest is so large that the first-best decision is not self-enforcing. At these states, self-enforcement requires that second-best decision-making be distorted in favor of the receiver. If the sender's and receiver's preferred decisions respond similarly to information, then first-best and second-best decision-making also respond similarly to information, so full separation is optimal. In contrast, if the receiver is highly responsive to information relative to the sender, then second-best decision-making is too responsive to information, so high-conflict states are optimally pooled to moderate second-best decision-making.
In sum, optimal relational communication involves pooling if and only if high-conflict states exist and the receiver is highly responsive.

The result that the sender hides information only when conflict of interest is sufficiently large seems to be a natural pattern of communication in relationships. Advocates often discuss in detail the costs and benefits of potential legislation with politicians, but may hide their private information in cases that are particularly controversial or consequential. 
Similarly, principals provide honest advice and agents comply when their preferences are largely aligned, but principals may hide information when agents are most tempted to dissent or disobey.

The result that pooling of high-conflict states occurs in relationships {only if the receiver is highly responsive} can be illustrated by fleshing out and  comparing our political advocacy and moviemaking examples. Starting with our advocacy example, suppose that the decision is how much to deregulate gun control and that the state is the degree of deregulation that is most popular with voters. Further, suppose that the sender/advocate is an ideologue who puts little weight on voter popularity and always prefers more deregulation than does the receiver/politician, while the politician is electorally-motivated and simply prefers the most popular policy. This corresponds to a receiver who is highly responsive (that is, his preferred policy is much more responsive to the state than that of the sender). Our model then predicts that the advocate will pool states where the conflict of interest is large, that is, will withhold information from the politician when deregulation is most unpopular. 

Turning to our moviemaking example, suppose that the decision is where to position a movie project's type along the arthouse--mainstream spectrum and that the state is the most profitable movie type. Further, suppose that the director/receiver puts little weight on profitability and always prefers a more art-house type than does the studio head/sender, while the studio simply prefers the most profitable type. This corresponds to a receiver who is not highly responsive. Our model then predicts that no pooling takes place: the studio will always fully inform the director, even in states where art-house projects would be very unprofitable and so the conflict of interest is large -- unlike in the political advocacy example.

In our model, pooling does not only occur at high-conflict states, when the receiver is highly responsive. Suppose the players are neither too patient nor too impatient, so that high-conflict and low-conflict states coexist. Then over-pooling occurs: high-conflict states are optimally pooled with some adjacent low-conflict states to further ease self-enforcement at those high-conflict states. In other words, optimal relationships hide information about some states where full separation and first-best decision-making could be supported in equilibrium. 

We also show that relational communication becomes more informative as the discount factor increases. As the players become more patient, second-best decision-making more closely approximates first-best decision-making and thus makes better use of information. Consequently, the sender optimally reveals more information to the receiver.

An implication of our analysis is that in settings where voluntary transfers are available, incomplete information transmission does not imply a failure to motivate communication, but instead is a tool to discipline decision-making. In other words, the Pareto frontier cannot be expanded simply by introducing a technology for credible (monotone) communication.\footnote{This is in contrast with the existing literature on cheap talk and delegation, where the receiver's expected payoff (which is the standard welfare criterion) unambiguously improves if credible communication can be costlessly achieved.} Indeed, we show that adding public information generally worsens the relationship.

\subsection{Related Literature}

In our model, transfers from the sender to the receiver are used to signal information. \cite{austen2000cheap} and \cite{kartik2007note} consider a related (albeit static) setting where the sender burns money to signal information.\footnote{\cite{kartik2007credulity} and \cite{kartik2009strategic} consider related models with lying costs instead of money burning.} Unlike burning money, signaling information with transfers incurs no welfare cost. This leads to a clean characterization of the set of optimal equilibria; in particular, all optimal equilibria in our model produce identical communication outcomes. In the setting with burned money, equilibrium communication outcomes differ along the Pareto frontier because there is a tradeoff between the efficiency of informed decision-making and the costs of burning money.\footnote{The receiver's optimal equilibrium clearly involves full separation; \cite{karamychevoptimal} characterize the sender's optimal equilibrium.} As a byproduct, we establish a general characterization of equilibria in games of cheap talk with burned money: a message rule is implementable if and only if it is monotone.\footnote{Relatedly, \cite{ottaviani2000economics}, \cite{krishna2008contracting}, and \cite{ambrus2017delegation} consider communication games where contractible transfers from the receiver to the sender are used to elicit the sender's information, as in mechanism design. Due to a limited liability constraint, the receiver has to leave information rents to the sender and thus trades off the efficiency of informed decision making with the corresponding information rents, leading to information pooling. In contrast, in our setting, information rents do not arise and information pooling occurs for Bayesian persuasion purposes.}

Our analysis builds on an extensive literature on repeated interactions with transfers. The seminal papers by \cite{Bull:1987db} and \cite{Macleod:1989kd} focus on settings with symmetric information. \cite{levin2003relational} characterizes the optimal relational contract in two important settings with asymmetric information: adverse selection and moral hazard. In these settings, only the decision-maker (agent) has private information. In contrast, our setting involves an informed sender and an uninformed decision-maker (receiver), in the vein of \cite{crawford1982strategic}. In such relational communication, pooling affects the receiver's beliefs and thus directly influences her decision-making. In contrast, the decision-maker (agent) in \cite{levin2003relational} is fully informed, so pooling has no such Bayesian persuasion effect. Instead, in \cite{levin2003relational}'s analysis of relational adverse selection, pooling of agent's types reduces the variability of transfers to satisfy the self-enforcement constraint.

\cite{alonso2007relational} also consider repeated communication. In contrast to us, they disallow transfers and consider a sequence of short-lived senders rather than a single long-lived sender.\footnote{\cite{baker2011relational} consider a model of repeated decision-making with transfers between long-lived players, but assume symmetric information, so communication plays no role.} In their setting, the sender may hide information when the receiver cannot credibly commit to take appropriate informed decisions. Repeated interactions allow the receiver to credibly sustain decisions favourable for the sender, so as to motivate the sender to communicate more information. In contrast, in our setting, credible communication can be costlessly achieved for any decision rule. Repeated interactions allow the receiver to credibly sustain decisions that improve joint surplus. If repeated interactions cannot sustain first-best decision-making, then communication is an additional tool to further improve decision outcome via Bayesian persuasion.\footnote{A model of repeated Bayesian persuasion would reproduce many of the insights from our model of relational communication. There is a literature on dynamic Bayesian persuasion, albeit with persistent information (\citealt{kremer2014implementing}, \citealt{au2015dynamic}, \citealt{ely2015suspense},  \citealt{horner2016selling}, \citealt{ely2016beeps}, \citealt{che2015optimal}, \citealt{ely2017moving}, \citealt{BRV}, \citealt{orlov2016persuading}, \citealt{best2017persuasion}, and \citealt{smolin2018dynamic}).}

In our model, optimal equilibria are supported by carrot-and-stick strategies (\citealt{abreu1986extremal} and \citealt{goldlucke2012infinitely}), in which a deviator is punished as harshly as possible but only for a single period. We show that the receiver is punished by complete pooling of information and the sender is punished by an extreme incentive compatible decision. These punishments also characterize the receiver's and sender's worst equilibria in games of cheap talk with burned money.

Our paper also contributes to the rapidly growing literature on Bayesian persuasion with transferable utility (\citealt{bergemann2007information}, \citealt{esHo2007optimal}, \citealt{li2017discriminatory}, \citealt{bergemann2018design}, and \citealt{dworczak2017mechanism}). Similarly to these papers, we use tools from mechanism design and Bayesian persuasion. Unlike these papers, commitment power in our model is endogenous and thus imperfect. 

\section{Model\label{Model}}
\subsection{Setup}\label{setup}
A sender ($S$) and a receiver ($R$) play an infinitely repeated communication game with perfect monitoring and with voluntary transfers. Time is discrete and the players have a common discount factor $\delta \in [0,1)$. In each period, the same stage game is played. The sender privately observes a state $\theta \in [0,1] $ and sends a message $m\subset[0,1]$ to the receiver, who then makes a decision $d\in \mathbb{R}$. The state $\theta $ is independently drawn each period from a prior distribution $F(\theta)$ with a strictly positive density $f(\theta)$ for all $\theta \in [ 0,1] $. The sender's and receiver's payoffs, $u_{S}\left( d,\theta \right)$ and $u_{R}\left( d,\theta \right)$, satisfy \cite{crawford1982strategic}'s assumptions:
\begin{assumption}\label{CS}\quad
For each player $i\in \{S,R\}$,
\begin{enumerate}
	\item $u_i(d,\theta)$ is twice differentiable in $d$ and $\theta$ for all $d\in \mathbb{R}$ and $\theta \in [0,1]$,
	\item $\frac{\partial^2 u_i}{\partial d^2}(d,\theta)<0$ for all $d\in \mathbb{R}$ and $\theta \in [0,1]$,
	\item $\frac{\partial u_i}{\partial d}(\rho_i(\theta),\theta)=0$ for some function $\rho_i(\theta)$ and for all $\theta\in [0,1]$,
	\item $\frac{\partial^2 u_i}{\partial d \partial \theta}(d,\theta)>0$ for all $d\in \mathbb{R}$ and $\theta \in [0,1]$.
\end{enumerate}	
\end{assumption}
Parts 2 and 3 of Assumption \ref{CS} require that each player's payoff is strictly concave in the decision and each player $i\in \{S,R\}$ has a unique \emph{preferred} decision $\rho_i(\theta)$ for each state $\theta\in [0,1]$. Similarly, there is a unique first-best decision $\rho_{FB}(\theta)$ that maximizes the joint payoff $u(d,\theta)=u_S (d,\theta)+u_R(d,\theta)$. Part 4 is a sorting condition that ensures that $\rho_S(\theta)$, $\rho_R(\theta)$, and $\rho_{FB}(\theta)$ are strictly increasing in $\theta$.

The players can make voluntary (non-contractible) transfers at any point in the game. Specifically, we enrich the stage game with three rounds of transfers: (i) an \textit{ex-ante} transfer $\tau \in \mathbb{R}$ before the sender observes the state, (ii) an \textit{interim} transfer $t \in \mathbb{R}$ after the sender observes the state and sends the message but before the receiver chooses a decision, and (iii) an \textit{ex-post} transfer $T \in \mathbb{R}$ after the decision is chosen.\footnote{Because next period's ex-ante transfer can substitute for this period's ex-post transfer, the set of equilibrium payoffs would not change if ex-post transfers were removed. But the analysis is simpler if ex-post transfers are allowed.} At each round, a positive transfer represents a payment from sender to receiver; conversely, a negative transfer represents a payment from receiver to sender. Thus, the stage game payoff of the sender is $u_{S}\left( d,\theta \right) -\tau-t-T$, and the stage game payoff of the receiver is $u_{R}\left( d,\theta \right) +\tau + t +T$. Since transfers are voluntary, the sender can reject a positive transfer, and the receiver can reject a negative transfer.

The game has perfect monitoring in that all actions (message, decision, and transfers) are immediately publicly observed, but the state is only observed by the sender. That is, the receiver never observes the state or her payoff.\footnote{This assumption is common in the literature on repeated games with incomplete information (\citealt{aumann1995repeated}), and is ubiquitous in models of repeated communication (\citealt{renault2013dynamic}, \citealt{frankel2016discounted}, \citealt{margaria2017dynamic}, and \citealt{lipnowski2017repeated}). In Section \ref{upward}, we briefly discuss the case where the state is publicly observed at the end of each period.} Figure \ref{Fig:timeline} summarizes the timing of each stage game.

	\begin{figure}[t!]
	\centering 
	\includegraphics[scale=0.6]
	{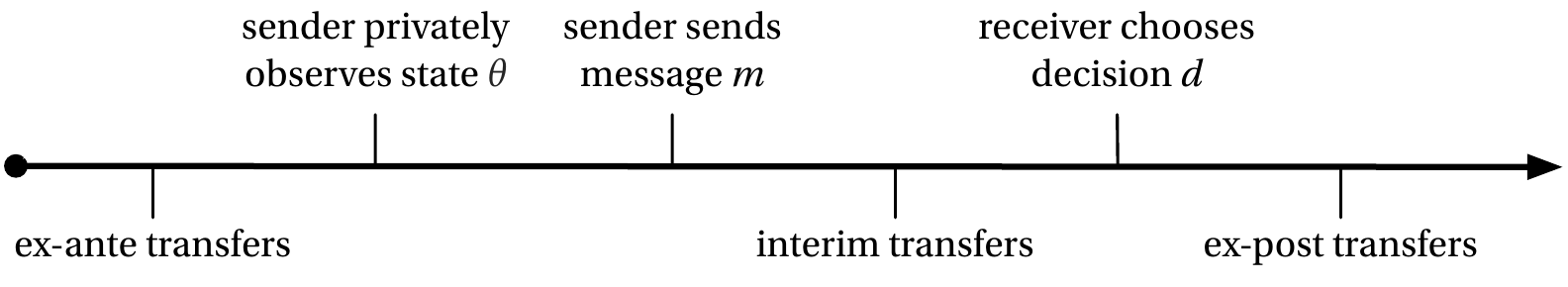}
	\caption{Timing of stage game
	} 
	\label{Fig:timeline}
	\end{figure}

We study pure-strategy perfect Bayesian \textit{equilibria}; later, in Footnote \ref{footnote:mixed}, we comment on mixed-strategy equilibria. For each period and each history, an equilibrium specifies (on-path) a message rule ${\mu}(\theta)$ for the sender, a decision rule $\rho(m)$ for the receiver, and transfer rules $\tau$, $t(m)$, $T(m)$ .\footnote{The functions $\mu$, $\rho$, $\tau$, $t$, and $T$ are required to be measurable.}

\bigskip

\emph{Conventions.} A (pure-strategy) message rule deterministically maps states to the messages they induce. Without loss of generality, we identify each message with the set of states that induce this message, $m=\{\theta:\mu(\theta)=m\}$. Thus, the range $\mu([0,1])$ of a message rule $\mu$ is a partition of the set of states. A message rule $\mu$ is \emph{monotone} if each $m\in \mu([0,1])$ is a convex set (either a singleton or an interval).\footnote{For example, a message rule that separately pools low states $[0,1/3)$ and high states $(2/3,1]$ into two distinct messages while separating intermediate states $[1/3,2/3]$ is monotone; but a message rule that pools low and high states into a single non-convex message $[0,1/3) \cup (2/3,1]$ while separating intermediate states is not monotone.}

We can now extend the definition of payoffs and preferred decisions from being state dependent to being message dependent. Specifically, $u_i(d,m)=\mathbb{E}_F[u_i(d,\theta)|m]$ and $\rho_i(m)=\arg\max_{d\in \mathbb R} u_i(d,m)$ for each player $i\in\{S,R\}$. Similarly, $u(d,m)=u_S (d,m)+u_R(d,m)$ and $\rho_{FB}(m)=\arg\max_{d\in \mathbb{R}}u(d,m)$. Assumption~\ref{CS} ensures that $\rho_S(m)$, $\rho_R(m)$ and $\rho_{FB}(m)$ are well defined and are strictly increasing in $m$ in the strong set order.  

\subsection{Stationarity}\label{sec:stationary}

We focus on stationary equilibria. An equilibrium is \emph{stationary} if on the equilibrium path, the message rule $\mu$, the decision rule $\rho$, and the transfer rules $\tau$, $t$, and $T$ are identical in every period. An equilibrium is \emph{optimal} if it is not Pareto dominated by any other equilibrium. 
An equilibrium is \emph{sequentially optimal} if the continuation equilibrium following any history on the equilibrium path is optimal.

\begin{lemma}\label{lem:stationary}
There exist $\underline{v}_{S}$, $\underline{v}_{R}$, and $\overline{v}$ such that the set of optimal equilibrium payoffs is the line segment
\begin{equation}
\overline V= \left\{ \left( v_{S},v_{R}\right)\in \mathbb R^2\, :\,  v_{S}\geq \underline{v}_{S},v_{R}\geq 
\underline{v}_{R},v_{S}+v_{R} = \overline{v}\right\} .  \label{simplex}	
\end{equation}
Any optimal equilibrium is sequentially optimal.  Further, there exists a stationary optimal equilibrium $\sigma_*$ such that any optimal equilibrium payoff vector $(v_S,v_R)$ can be supported by an equilibrium that differs from $\sigma_*$ only in the first-period ex-ante transfer.
\end{lemma}

Lemma \ref{lem:stationary} extends some of \cite{levin2003relational}'s and \cite{goldlucke2012infinitely}'s results to our setting, with an extensive-form stage game of incomplete information. Because players' payoffs are quasi-linear in money, payoffs are fully transferable, and contingent transfers can substitute for contingent continuation payoffs. Consequently, we can restrict attention to stationary equilibria, and all optimal equilibria induce the message and decision rules that maximize joint payoff $v=v_S + v_R$.

\section{Equilibrium\label{Analysis}}
\subsection{Implementability\label{sec:impl}}

We now show that the presence of interim and ex post voluntary transfers enables separation of the sender's and receiver's incentive constraints. 
The sender's incentive constraint requires that the decision outcome be monotone.
The receiver's incentive constraint requires that induced decisions be close to the receiver's preferred decisions.

Define the receiver's temptation to deviate from decision $d$ given message $m$ as
\begin{equation*}
w(d,m)= u_{R}(\rho_R(m),m )-u_R(d,m),
\end{equation*}
and the \emph{net discounted surplus} given joint payoff $v$ as
\begin{equation*}\label{eq:L}
L(v)= \frac{\delta}{1-\delta}(v-\underline{v}_S-\underline{v}_R).	
\end{equation*}

\begin{proposition}\label{pr:impl}
A message rule $\mu$ and a decision rule $\rho$ that produce a joint payoff $v$ can be supported in a stationary equilibrium if and only if the decision outcome is monotone:
\begin{gather}
\rho(\mu(\theta))\text{ is non-decreasing in }\theta,\label{IC}
\end{gather}
and the receiver's temptation to deviate never exceeds the net discounted surplus:
\begin{gather}
	w(\rho(m),m)\leq L(v)\text{ for all }m\in \mu([0,1]).\label{enforcement}
\end{gather}
\end{proposition}

We first argue that (\ref{IC}) and (\ref{enforcement}) are necessary. In any equilibrium, the message rule $\mu(\theta)$ must be incentive compatible for the sender.  Since the sender's payoff is quasi-linear in money and satisfies a sorting condition, a standard characterization of incentive compatibility in mechanism design (see, for example, \citealt{rochet1987necessary}) implies that $\rho(\mu(\theta))$ must be non-decreasing in $\theta$.

Also, in any equilibrium, the decision rule $\rho$ must be incentive compatible for the receiver. Therefore, given a message $m$, the receiver's one-period payoff gain from choosing her preferred decision $\rho_R(m)$ instead of equilibrium decision $\rho(m)$ must be less than the maximum available punishment equal to the discounted surplus.

We now argue that (\ref{IC}) and (\ref{enforcement}) are sufficient. Ignoring the sender's incentive compatibility constraint, any decision rule $\rho$ that satisfies (\ref{enforcement}) can be made incentive compatible for the receiver by giving all surplus to the receiver ($v_R = v - \underline{v}_S$) and threatening her with her worst equilibrium payoff ($v_R = \underline{v}_R$) following any deviation from $\rho(m)$. 

In such a construction, the sender receives his worst equilibrium payoff $\underline{v}_S$ and thus cannot be punished for deviating. But for any message rule $\mu$ that satisfies (\ref{IC}), we can separately construct a (voluntary) interim transfer rule that makes $\mu$ incentive compatible for the sender.

The envelope theorem (see, for example, \citealt{milgrom2002envelope}) implies that there exists a unique (up to a constant $C$) interim transfer rule $t$ such that the sender prefers to induce $\rho(\mu(\theta))$ and pay $t(\mu(\theta))$ rather than to induce $\rho(\mu(\hat{\theta}))$ and pay $t(\mu(\hat{\theta}))$ for all $\hat\theta \neq \theta$,
\begin{align}
t(m)=u_S(\rho(m),\theta(m))-\int_0^{\theta(m)} \frac{\partial u_S}{\partial \theta} (\rho(\mu(\tilde\theta)),\tilde\theta) d\tilde\theta +C, \label{eq:h}
\end{align}
where $\theta(m)$ is an arbitrary state $\theta\in m$.\footnote{Since $\rho(\mu(\theta))$ is non-decreasing in $\theta$, $t(m)$ is independent of the choice of a representative state $\theta\in m$.} The constant $C$ can be chosen in such a way that the sender does not want to deviate to any out-of-equilibrium message-transfer pair $(\hat{m},\hat{t})$. Specifically, choose $C$ such that the minimum transfer is equal to zero and is achieved for some punishment message $m^P$,\footnote{In the proof, we allow for the possibility that $\inf t(m)$ is not attained by any $m^P$.}
\begin{equation}\label{eq:ts}
t(m) \geq 0\text{ for all }m\in \mu([0,1]),\text{ with equality for some }m^P\in \mu([0,1]).
\end{equation}
If following any out-of-equilibrium pair $(\hat{m},\hat{t})$, the receiver believes that the state is in $m^P$ and chooses the punishment decision $d^P=\rho(m^P)$, then the sender prefers to report $m^P$ and pay $t(m^P)=0$ rather than to report $\hat{m}$ and pay $\hat{t}$. Thus, the sender's incentive constraint is satisfied.\footnote{Note that self enforcement does not impose any limit on the variability of interim transfers and thus does not create a shadow cost to screening, which is a key driving force in \cite{levin2003relational}'s analysis of relational adverse selection.}

This argument implies that voluntary interim transfers are powerful in signaling information, even if the players are myopic.\footnote{Although interim transfers are powerful, messages are still used to convey information. For example, suppose the players' preferred decision rules intersect at some state. Then in any fully separating equilibrium, the interim transfer function is non-monotone and takes the same value for multiple state realizations. Messages are thus used to distinguish between these realizations.}

\begin{corollary}\label{cr:myopic}
Suppose $\delta=0$. A message rule $\mu$ and a decision rule $\rho$ can be supported in an equilibrium if and only if $\mu$ is monotone and $\rho(m)=\rho_R(m)$ for all $m\in \mu([0,1])$.
\end{corollary}

Corollary~\ref{cr:myopic} is closely connected to existing results from the literature on cheap talk and burned money (\citealt{austen2000cheap}, \citealt{kartik2007note}, and \citealt{karamychevoptimal}). In the myopic setting, interim transfers serve the same signaling role as burned money. In fact, the set of implementable message and decision rules does not depend on whether the sender transfers money to the receiver or whether the sender burns money.\footnote{\cite{karamychevoptimal}'s Proposition 1 characterizes implementable outcomes with money burning. Our mechanism design approach to characterization provides a much simpler proof of the result and removes the assumptions that the bias $\rho_R-\rho_{FB}$ has constant sign and that the receiver's payoff satisfies a sorting condition. Indeed, if the receiver's payoff did not satisfy part 4 of Assumption \ref{CS}, our Proposition \ref{pr:impl} and its proof would still hold, but Corollary~\ref{cr:myopic} would require that $\rho_R(\mu(\theta))$ be non-decreasing in $\theta$, rather than that $\mu$ be monotone.}

In contrast to burned money, interim transfers are not wasteful: the sender's loss is the receiver's gain. Further, since ex-ante transfers are available, the use of interim transfers does not create a distributional imbalance. Any surplus obtained by the receiver from interim transfers can be redistributed to the sender using ex-ante transfers. Such ex-ante transfers can be supported by the threat of complete pooling of information. Consequently, the sender can commit at no welfare cost to any monotone message rule.

\subsection{Optimality}\label{sec:opt}
An optimal equilibrium solves a monotone persuasion problem: it maximizes the expected joint payoff over monotone message rules, subject to the second-best decision rule. 

Define the second-best decision given message $m$ as 
\begin{equation}
\begin{gathered}\label{dstar}
	\rho_*(m)=\arg \max_d u (d,m)\\
	\text{subject to }w(d,m)\leq L(\overline{v}),
\end{gathered}
\end{equation}
and the joint payoff under the second-best decision as
\begin{equation}\label{eqn:Gamma}
u_*(m) = u(\rho_*(m),m)\text{ for all }m \subset [0,1]. 	
\end{equation}
\begin{proposition}\label{pr:opt}
In a stationary optimal equilibrium, the message rule satisfies
\begin{equation}
\begin{gathered}\label{mustar}
\mu_*\in\arg\max_{\mu} \; \mathbb{E} \left[ u_*(\mu(\theta)) \right]	 \\
\text{subject to } \mu \text{ is monotone,}
\end{gathered}
\end{equation}
and the second-best decision $\rho_*(m)$ is taken for each on-path message $m\in \mu_*([0,1])$.
\end{proposition}

The intuition for Proposition~\ref{pr:opt} is as follows. By Proposition~\ref{pr:impl}, an optimal equilibrium maximizes $v$ jointly over message and decision rules that satisfy the sender's and receiver's incentive constraints, (\ref{IC}) and (\ref{enforcement}). By constraint~(\ref{IC}) and a revelation principle argument, we can restrict attention to monotone message rules. Consider a relaxed problem in which the constraint (\ref{IC}) is replaced with the constraint that the message rule is monotone. It is easy to see that $\rho_*$ given by (\ref{dstar}) and $\mu_*$ given by (\ref{mustar}) solve this relaxed problem. Further, we show that $\rho_*(m)$ is non-decreasing in $m$ because the sender's and receiver's payoffs satisfy the sorting condition (part 4 of Assumption \ref{CS}). Therefore, $\rho_*(\mu_*(\theta))$ is non-decreasing in $\theta$, the constraint (\ref{IC}) is automatically satisfied, and $\rho_*$ and $\mu_*$ constitute an optimal equilibrium.

Proposition~\ref{pr:opt} shows that the decision rule and message rule in any optimal equilibrium can be calculated in two steps. First, the decision rule is characterized without reference to the message rule. The decision rule is point-wise equal to the second-best decision $\rho_*(m)$ given by~(\ref{dstar}). For each message $m$, the second-best decision $\rho_*(m)$ can be found as follows. If $d=\rho_{FB}(m)$ satisfies the constraint of (\ref{dstar}), then $\rho_*(m)=\rho_{FB}(m)$. Otherwise $\rho_*(m)$ is the unique decision $d$ that lies between $\rho_R(m)$ and $\rho_{FB}(m)$ and satisfies the constraint of (\ref{dstar}) with equality.
Second, given $\rho_*$ and thus $u_*$, the message rule $\mu_*$ solves the \emph{monotone persuasion} problem (\ref{mustar}): it maximizes the expected joint payoff $\mathbb{E}[u_*(\mu(\theta))]$ over all monotone message rules $\mu$.\footnote{\label{footnote:mixed}We have restricted attention to pure-strategy equilibria. Proposition~\ref{pr:opt} would continue to hold if we instead restricted attention to equilibria where the receiver uses a pure strategy. Indeed, since $u_S$ is supermodular, in any such equilibrium, the sender's incentive constraint implies that the decision outcome must be monotone, with randomizations only at a countable set of states which -- since $\theta$ has a density -- do not affect expected payoffs. Moreover, Proposition~\ref{pr:opt} would continue to hold without restriction to pure strategies for either player if we additionally assume that $\partial ^2 u_S(d,\theta)/\partial d\partial \theta=\beta'(\theta)\gamma'(d)$ for some increasing functions $\beta$ and $\gamma$ -- as, for example, in Assumption \ref{quad}. Under this additional assumption, in any mixed-strategy equilibrium, the sender's incentive constraint implies that a higher state induces a higher expectation of $\gamma(d)$, with randomizations by the sender only at a countable set of states. Thus, as far as payoffs are concerned, we can restrict attention to equilibria with a pure monotone message rule. Given this restriction and the assumption that payoffs are concave in the decision, the second-best decision rule is pure and so is the receiver's strategy in any optimal equilibrium.}

\subsection{High-Conflict States}\label{sec:conflict}
We say that a state $\theta$ is \emph{high-conflict} if the first-best decision is not enforceable at this state, $w(\rho_{FB}(\theta),\theta)>L(\overline{v})$; otherwise, the state is \emph{low-conflict}.
\begin{corollary}\label{c:separation_necessary}
If full separation is suboptimal, then some states are high-conflict, and each non-singleton message $m \in \mu_*([0,1])$ contains some high-conflict states.
\end{corollary}
\begin{proof}
Consider the second-best decision rule and a monotone message rule where some non-singleton message $m$ consists only of low-conflict states. Then the expected joint payoff can be increased by separating all states in $m$, thus implementing the first-best decision $\rho_{FB}(\theta)$ for each $\theta\in m$; while keeping the other messages unchanged.
\end{proof}

Further, we can specify sufficient local conditions for pooling to be optimal at some high-conflict states.
\begin{proposition}\label{prop:localpool}
Suppose $L(\overline{v})>0$, and $u_S(d,\theta)$ and $u_R(d,\theta)$ are thrice differentiable in $d$ and $\theta$. Full separation is suboptimal if there exists a high-conflict state $\hat\theta\in [0,1]$ such that
\begin{gather}
\frac{1}{2}{u''_{dd}(\rho_*( \hat\theta),\hat\theta)}\left( \rho'_*(\hat\theta)-2 \tilde \rho_{FB}'(\hat\theta)  \right)	+ {u'_d(\rho_*( \hat\theta),\hat\theta)}\left( \frac{\rho_*''(\hat\theta)}{2 \rho_*'(\hat\theta)}  -  \frac{\varphi''(\hat\theta)}{2 \varphi'(\hat \theta)}  \right)<0	
	\label{eqn:rhoprime} \\
\text{where} \quad \tilde\rho_{FB}'(\hat\theta)=-\frac{u''_{d\theta}(\rho_*(\hat\theta),\hat\theta)}{u''_{dd}(\rho_*( \hat\theta),\hat\theta)}\quad \text{and}\quad \varphi(\theta)=u_R(\rho_R(\hat\theta),\theta)-u_R(\rho_*(\hat\theta),\theta). \nonumber
\end{gather}
\end{proposition}

Separating a high-conflict state $\hat\theta$ is suboptimal if for a small neighborhood $m$ around $\hat\theta$, pooling is better than separation, 
$\mathbb E[u(\rho^*(m),\theta)| m]>\mathbb E[u(\rho^*(\theta),\theta)|m]$. Equation \eqref{eqn:rhoprime} shows that the tradeoff between separation and pooling is affected by the slope of $\rho_*$ and the curvatures of $\rho_*$ and $\varphi$. 

Importantly, as it turns out, the second-best decision rule can be approximated as 
\begin{equation}
\rho_*(m)\approx \rho_* \left(\varphi^{-1}\left(\mathbb E[\varphi(\theta) |m]\right)\right).\label{eqn:rho_approx}
\end{equation}
To see why, notice that the receiver's binding incentive constraint $w(\rho_*(m),m)=L(\overline{v})$ implies
\[\rho_*(m)-\rho_*(\hat\theta) \approx \frac{\mathbb E[u_R(\rho_R(\hat\theta),\theta)-u_R(\rho_*(\hat\theta),\theta)| m]}{{u_R}'_d(\rho_*(\hat\theta),\hat\theta)}=\frac{\mathbb{E}[\varphi(\theta) |m]}{{u_R}'_d(\rho_*(\hat\theta),\hat\theta)}.\]
That is, locally, $\rho_*(m)$ depends on $m$ only through $\mathbb E[\varphi(\theta)|m]$, so we can approximately express $\rho_*(m)\approx \zeta (\mathbb E[\varphi(\theta) |m])$ for some function $\zeta$. Expression (\ref{eqn:rho_approx}) then follows from the observation that $\zeta(\cdot)=\rho_*(\varphi^{-1}(\cdot))$, which in turn follows from $\rho_*(\theta)=\zeta(\varphi(\theta))$. 

The curvatures of $\rho^*$ and $\varphi$ jointly determine the expected second-best decision under pooling, $\rho_* (\varphi^{-1}(\mathbb E[\varphi(\theta) |m]))$, and separation, $\mathbb E[  \rho_* (\varphi^{-1}(\varphi(\theta)))|m]$.
Comparing the expected decision under pooling and separation is equivalent to comparing the expected utility of a consumer with a Bernoulli utility function  $\rho_* (\varphi^{-1}(\cdot))$ under two fair lotteries over consumption quantities. Pooling corresponds to a degenerate lottery that yields a certain quantity $\mathbb{E}[\varphi(\theta) |m]$ and separation corresponds to a risky lottery that yields a random quantity $\varphi(\theta)|m$. Consequently, as $\rho_*$ and $\varphi^{-1}$ become more concave, the consumer becomes more risk-averse -- which makes the degenerate lottery more attractive, so the expected decision becomes relatively higher under pooling. Thus, pooling is more favorable if the first-best decision is higher than the second-best decision, $\rho_{FB}(\hat\theta)>\rho_*(\hat \theta)$, or equivalently $u'_d(\rho_*( \hat\theta),\hat\theta)>0$.

The slope of $\rho_*$ does not affect the expected decisions under pooling and separation, but instead determines whether pooling or separation better approximates the slope of the `locally first-best' decision rule.\footnote{\label{f:locfb}The `locally first-best' decision rule $\tilde\rho_{FB}(\theta)$ maximizes $\mathbb E[u(\rho(\theta),\theta)|m]$ amongst all linear rules $\rho(\theta)$ that satisfy $\rho(\hat\theta)=\rho_*(\hat \theta)$.} The decision outcome has slope $\rho_*'(\hat\theta)$ under separation and slope 0 under pooling.  Pooling is thus more favorable if $\tilde \rho_{FB}'(\hat\theta)$ is closer to $0$ than to $\rho_*'(\hat\theta)$, or equivalently $\rho_*'(\hat \theta)>2\tilde\rho_{FB}'(\hat\theta)$. We will develop this intuition more fully in Section \ref{sec:QuadOpt}, where the payoffs are quadratic.

In the myopic case, where $L(\overline{v})=0$ and correspondingly $\varphi(\theta)={u_R}'_d(\rho_R(\hat\theta),\theta)$, our setting reduces to that of the monotone persuasion problem in which the decision and the state are both one-dimensional. Proposition \ref{prop:localpool} can thus be interpreted as specifying local sufficient conditions for optimal persuasion to involve some pooling. Closest to this analysis is \cite{jehiel2015transparency}, who also consider local sufficient conditions for optimal persuasion to involve pooling. His main result is that full separation is suboptimal if there exist two distinct states at which the receiver's preferred decision is the same. This result has no bite in our setting, where the receiver's preferred decision is strictly increasing in the state.

To summarize, Corollary \ref{c:separation_necessary} provides necessary conditions for pooling to be optimal, whereas Proposition \ref{prop:localpool} provides sufficient conditions. The next section specializes to the case of quadratic payoffs, so as to provide conditions that are both necessary and sufficient for pooling to be optimal and to fully characterize optimal relational communication.

\section{Quadratic Payoffs}\label{sec:quadratic}

\begin{assumption}\label{quad}
The sender's and receiver's payoffs are given by $u_S(d,\theta)=\lambda_S(\rho_S(\theta) d-d^2/2)$ and $u_R(d,\theta)=\lambda_R(\rho_R(\theta) d -d^2/2)$ for all $d\in \mathbb{R}$ and $\theta\in [0,1]$ where preferred decisions $\rho_S(\theta)$ and $\rho_R(\theta)$ are increasing and linear in $\theta$, and weights $\lambda_S>0$ and $\lambda_R>0$ satisfy $\lambda_S+\lambda_R=1$.
\end{assumption}
Assumption \ref{quad} implies Assumption \ref{CS}, so all our previous results apply.  Given the normalization $\lambda_S+\lambda_R=1$, the first-best decision is $\rho_{FB}(\theta)= \lambda_R \rho_R(\theta) + \lambda_S \rho_S(\theta)$ and the joint payoff is $u(d,\theta)=\rho_{FB}(\theta)d-d^2/2$.\footnote{These payoff functions nest two special cases. First, in \cite{crawford1982strategic}'s example, the sender has constant upward bias, so that $\rho_S(\theta)=\rho_R(\theta)+b$ with $b>0$. Second, in \cite{kamenica2011bayesian}'s lobbying example, the sender is biased toward a specific decision $d_c\in \mathbb{R}$, so that $\rho_S(\theta)=a  \rho_R(\theta) + (1-a) d_c$ with $a\in(0,1)$.}

Assumption \ref{quad} ensures that $u_*(m)$ depends on a message only through the induced posterior mean state: $u_*(m)=u_*(\mathbb{E}[\theta|m])$ for all $m \subset [0,1]$.\footnote{Bayesian persuasion in settings where payoffs depend only on the induced posterior mean state has been studied, for example, by \cite{gentzkow2016rothschild}, \cite{kolotilin2017persuasion}, \cite{kolotilin2014informed}, and \cite{dworczak2017simple}.} Therefore, without loss of generality, we identify each message $m$ with the induced posterior mean state, $m=\mathbb{E}[\theta|m]$. This simplifies the previous convention that identified each message with the set of states that induce it.

\subsection{Pooling versus Separation}\label{sec:QuadOpt}

We first show that it is optimal to pool some states whenever the players are not patient enough to enforce the first-best outcome and the slope of the receiver's preferred decision rule is sufficiently high.

Under Assumption \ref{quad}, the second-best decision (\ref{dstar}) given message $m$ pushes $d$ as close to $\rho_{FB}(m)$ as possible, while still keeping $d$ within distance $\ell$ from the receiver's preferred decision $\rho_R(m)$:

\begin{equation*}
\begin{gathered}\label{dast}
	\rho_*(m)=\arg \max_d u (d,m)\\
	\text{subject to }\left\vert d -\rho_R(m)\right\vert \leq \ell =\sqrt{\frac{L(\overline v)}{\lambda_R}},
\end{gathered}
\end{equation*}
where we call $\ell$ the relational \emph{leeway}. The second-best decision rule $\rho_*$ is parallel to $\rho_R$ at high-conflict states, where $|\rho_{FB}(\theta)-\rho_R(\theta)|>\ell$, and coincides with $\rho_{FB}$ at low-conflict states.

The tradeoff between pooling and separation is tightly linked to the curvature of the joint payoff under the second-best decision
\[u_*(\theta)=\rho_{FB}(\theta)\rho_*(\theta)-\frac{\rho_*^2(\theta)}{2},\]
which is in turn determined by how the second-best decision responds to the state.
Define the receiver to be \emph{highly responsive} if 

\begin{equation}\label{eqn:responsive}
    \rho'_R(\theta)>2\rho'_{FB}(\theta). 
\end{equation}

We can strengthen Corollary \ref{c:separation_necessary} and Proposition \ref{prop:localpool} as follows.

\begin{corollary}\label{c:separation}
Full separation is suboptimal if and only if some states are high-conflict and the receiver is highly responsive. In that case, $u_*$ is concave for high-conflict states and convex for low-conflict states.
\end{corollary}
\begin{proof}
As shown in \cite{kamenica2011bayesian}, if $u_*$ is convex, then full separation is optimal. Since each message rule $\mu$ is less informative than the fully separating message rule, the prior distribution $F$ is a mean-preserving spread of the distribution of posterior means induced by $\mu$. Thus, if $u_*$ is convex, then \[\mathbb{E} [u_*(\theta)] \geq \mathbb{E}[ u_*(\mu(\theta))],\] showing that full separation is optimal. 

Conversely, if $u_*$ is strictly concave on some non-empty interval $(\theta_1,\theta_2)$, then full separation is suboptimal. Indeed, if $u_*$ is strictly concave on $(\theta_1,\theta_2)$, then \[\mathbb{E} [u_*(\theta)|\theta\in (\theta_1,\theta_2)]<u_*(\mathbb{E}[\theta)|\theta\in (\theta_1,\theta_2)]),\] showing that full separation can be improved upon by pooling states $(\theta_1,\theta_2)$.

To complete the proof of the corollary, it suffices to show that $u_*$ is continuously differentiable and that $u_*''(\theta)<0$ if and only if the state $\theta$ is high-conflict and the receiver is highly responsive. This is established in Lemma \ref{l:d2u} in Appendix \ref{AppQuad}, with
\begin{align*}\label{d2u}
u_*''(\theta) &=(2\rho'_{FB}(\theta) -\rho'_*(\theta))\rho'_*(\theta) \\
&=
\begin{cases}
\rho'^2_{FB}(\theta), &\text{if $\theta$ is low-conflict,}\\
\left(2\rho'_{FB}(\theta)-\rho'_R(\theta)\right)\rho'_R(\theta),  &\text{if $\theta$ is high-conflict.}
\end{cases}\qedhere
\end{align*}
\end{proof}

Corollary \ref{c:separation} shows that at high-conflict states, condition (\ref{eqn:rhoprime}) simplifies to (\ref{eqn:responsive}), and further becomes necessary and sufficient for pooling to be optimal.\footnote{The curvatures of $\varphi$ and $\rho_*$ play no role in the separation-pooling tradeoff because  $\varphi$ is linear and $\rho_*$ is piecewise linear.} We pause to explain in detail how the receiver's responsiveness drives the separation-pooling tradeoff.
Suppose that all states are high-conflict and the receiver is downward-biased, so that the second-best decision rule is $\rho_*(m) = \rho_R(m)+\ell$. The optimal message rule should induce a decision outcome that approximates the first-best decision outcome as closely as possible.  

To understand how the message rule shapes decision-making, first notice that any message rule must induce receiver's beliefs that are correct in expectation, so that every message rule always induces the same expected decision outcome. Thus, the expected decision outcome cannot be moved closer to the expected first-best decision by changing the message rule. Optimizing the message rule thus involves making the `slope' of the induced decision outcome as close as possible to that of the first-best decision outcome.

Let us compare complete pooling and full separation. Complete pooling induces a completely `flat' decision outcome that is completely unresponsive to the state. Full separation induces a decision outcome that runs parallel to the receiver's preferred decision, so $\rho'_*(\theta)=\rho'_R(\theta)$. This decision outcome is clearly more responsive to the state than complete pooling, and further is more responsive than the first-best outcome if $\rho'_R(\theta)>\rho'_{FB}(\theta)$. Thus the first-best slope is better approximated by full separation than by complete pooling if and only if $\rho'_{FB}(\theta)$ is closer to $\rho'_R(\theta)$ than to $0$, or equivalently $\rho'_R(\theta) < 2 \rho'_{FB}(\theta)$; that is, the receiver is not highly responsive.

This logic extends to the more general case where not all states are high-conflict. It is optimal to fully separate states if the receiver is not highly responsive, and to pool the high-conflict states if the receiver is highly responsive.  We will next show how the high-conflict states should be optimally pooled with low-conflict states, when the receiver is highly responsive.

\subsection{Optimal Communication}\label{upward}
We characterize optimal equilibria in the case where the receiver is highly responsive and is downward-biased, $\rho_R(\theta)<\rho_{FB}(\theta)$ for all $\theta\in [0,1]$. Here, the set of high-conflict states is a (possibly empty) interval $[0,\theta_{\mathrm{c}})$ with $\theta_{\mathrm{c}} \in [0,1]$. At the end of this section, we will briefly discuss the case where the receiver is upward-biased for some states and downward-biased for others.

\begin{proposition}\label{pr:sb}
Suppose the receiver is highly responsive and downward-biased. There exists a cutoff state $\theta_*\in [0,1]$ such that the optimal message rule pools the states below $\theta_* $ into one message and separates the states above $\theta_* $. Furthermore, the pooling interval includes all high-conflict states and some adjacent low-conflict states, $\theta_* > \theta_{\mathrm{c}}$, if high-conflict and low-conflict states coexist, $\theta_{\mathrm{c}} \in (0,1)$. Otherwise, $\theta_{\mathrm{c}}\in  \{0,1\}$, the pooling interval is equal to the set of high-conflict states, $\theta_* = \theta_{\mathrm{c}}$.
\end{proposition}

\begin{proof}
By Corollary \ref{c:separation}, if $\theta_{\mathrm{c}}=0$, then $u_*$ is convex and full separation is optimal. Similarly, if $\theta_{\mathrm{c}}=1$, then $u_*$ is concave and complete pooling is optimal. Moreover, if $\theta_{\mathrm{c}}\in (0,1)$, then $u_*$ is concave on $[0,\theta_{\mathrm{c}}]$ and convex on $[\theta_{\mathrm{c}},1]$. Then Proposition \ref{pr:sb} follows from Remark 1 in \cite{kolotilin2019censorship}.\footnote{Related results also appear in \cite{kolotilin2017persuasion}, \cite{kolotilin2014informed}, and \cite{dworczak2017simple}.} For completeness, we outline the proof here. Consider an auxiliary payoff function \[
u_A(\theta)
=
\begin{cases}
u_*(m_*)+{u}_*'(m_*)(\theta- m_*), &\text{if $\theta<\theta_*$},\\
{u}_*(\theta), &\text{if $\theta\geq \theta_*$},
\end{cases}
\]
where $m_*=\mathbb E[\theta|\theta<\theta_*]$ and $\theta_*\geq \theta_{\mathrm{c}}$ is such that ${u}_A(\theta)$ is continuous at $\theta_*$ (Figure \ref{Fig:Aux}). Notice that $u_A$ is convex and $u_A\geq u_*$. Let us solve the auxiliary problem of choosing a message rule $\mu$ to maximize $\mathbb{E}[u_A(\mu(\theta))]$. Since  $u_A$ is convex, full separation is optimal in the auxiliary problem, and the maximum expected auxiliary payoff is
\begin{align*}
\overline v_A &= F(\theta_*) \mathbb E[u_A(\theta)|\theta<\theta_*] + (1-F(\theta_*)) \mathbb E[u_A(\theta)|\theta\geq\theta_*]\\
&= F(\theta_*) u_*(m_*) + (1-F(\theta_*)) \mathbb E[u_*(\theta)|\theta\geq \theta_*].
\end{align*}
Since $u_A \geq u_*$, $\overline v_A$ is an upper bound on the maximum expected joint payoff $\overline v$ in problem (\ref{mustar}). But the message rule that pools the states below $\theta_*$ and separates the rest achieves this upper bound. Since this message rule is monotone, it solves the monotone persuasion problem (\ref{mustar}).
\end{proof}
\begin{figure}[t!]
	\centering 
	\includegraphics[scale=0.6]
	{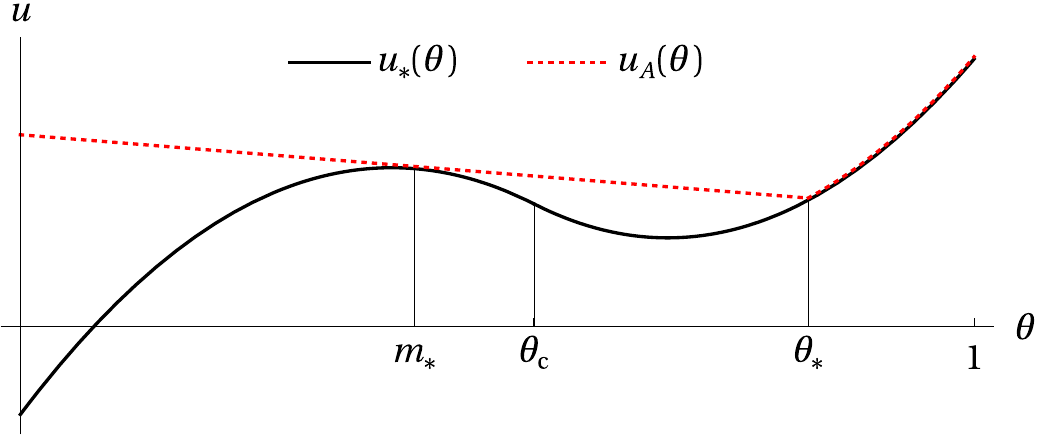}
	\caption{Joint and auxiliary payoff functions} 
	\label{Fig:Aux}
\end{figure}
Proposition \ref{pr:sb} highlights what we call \emph{over-pooling}: all high-conflict states are optimally pooled with some adjacent low-conflict states (Figure \ref{Fig:UpBias}). To see why over-pooling occurs, consider the effect of marginally expanding the pooling interval from $[0, \theta_*)$ to $[0,{\theta}_*+d{\theta})$. The cost of this expansion is that marginal states switch from the first-best decision to the second-best decision, for a loss of $\left(u_*(\theta_*)-u_*(m_*)\right) f(\theta_*) d \theta$. The benefit of this expansion is that decision-making for the pooled states marginally improves towards the first-best decision (Figure \ref{sub:Up2}), for a gain of $u'_*(m_*)(\theta_* - m_*) f(\theta_*) d\theta$. The cutoff state, if interior, is such that the benefit equals the cost:
\begin{equation}
u'_*(m_*)(\theta_* - m_*) = (u_*(\theta_*)-u_*(m_*)).	\label{eqn:cutoff}
\end{equation}
If the receiver is highly responsive, then $u_*$ is concave on $[0,\theta_{\mathrm{c}})$ (Figure \ref{Fig:Aux}). In this case, evaluation of (\ref{eqn:cutoff}) at $\theta_*=\theta_{\mathrm{c}}$ indicates that the benefit of the marginal expansion outweighs the cost, leading to over-pooling: the cutoff state $\theta_*$ is greater than $\theta_{\mathrm{c}}$. Intuitively, with a highly responsive receiver, the second-best decision rule -- which runs parallel to the receiver's preferred decision at high-conflict states -- is also highly responsive to the message. So, expanding the pool results in a large shift in the pooled decision towards the first-best, and thus a large benefit (relative to the cost) from over-pooling.

\begin{figure}[t!] 
	\centering 
	\includegraphics[scale=0.75]{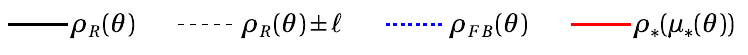}\\ \vspace{-1em}
	\subfloat[High $\delta$: Full Separation]
	{\label{sub:Up1}
	\includegraphics[scale=0.57]
	{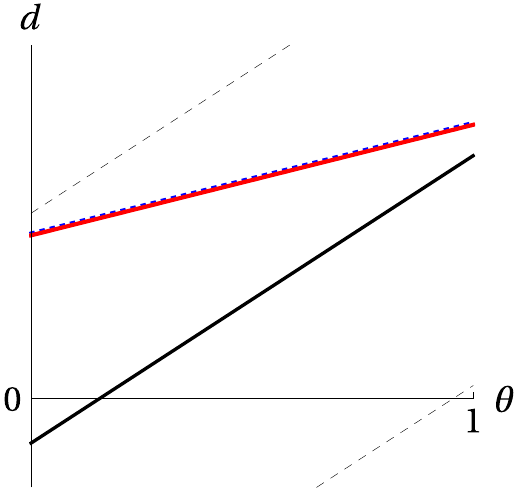}}
	\quad 
	\subfloat[Medium $\delta$: Partial pooling]
	{\label{sub:Up2}
	\includegraphics[scale=0.57]
	{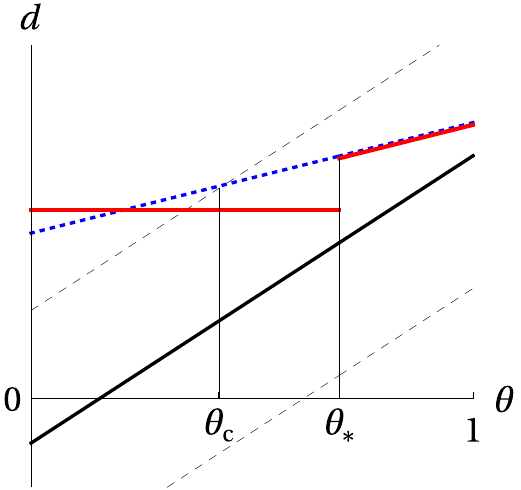}}
	\quad 
	\subfloat[Low $\delta$: Complete Pooling]
	{\label{sub:Up3}
	\includegraphics[scale=0.57]
	{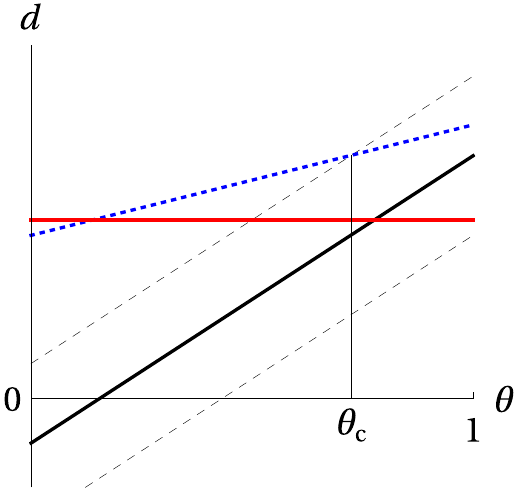}}	\\
    \medskip
	\begin{minipage}{.94\textwidth}
	{\tablenote  Optimal decision rules and decision outcomes given a highly-responsive, downward-biased receiver. Figure \ref{sub:Up1} shows the high-$\delta$ case, with full separation and $\theta_{\mathrm{c}}=\theta_*=0$; Figure \ref{sub:Up2} shows the intermediate-$\delta$ case, with partial pooling and $0 < \theta_{\mathrm{c}} < \theta_* < 1$; Figure \ref{sub:Up3} shows the low-$\delta$ case, with complete pooling and $0 < \theta_{\mathrm{c}} < \theta_* = 1$. The last two cases illustrate overpooling: $\theta_{\mathrm{c}} < \theta_*$.  \par}
	\end{minipage}	
	\caption{Optimal pooling with a highly-responsive, downward-biased receiver} 
	\label{Fig:UpBias}
\end{figure}

As the players become more patient, the interval of high-conflict states $[0, \theta_{\mathrm{c}})$ shrinks,\footnote{This is because the relational leeway $\ell$ increases with $\delta$, as shown in Lemma \ref{elldelta} in Appendix \ref{AppQuad}.} and the pooling interval shrinks with it.
\begin{corollary}\label{c:cs}
Suppose the receiver is highly responsive and downward-biased. As the players become more patient, the pooling interval shrinks, $d\theta_*/d\delta \leq 0$; strictly so if and only if $\theta_* \in (0,1)$.
\end{corollary}

Corollary \ref{c:cs} states that more information is optimally revealed in the sense of \cite{blackwell1953equivalent} as the players become more patient. With more patient players, the second-best decision rule more closely approximates the first-best decision rule and thus makes better use of information. Consequently, the sender optimally reveals more information.\footnote{Relatedly, as Corollary~\ref{c:csbeta} in Appendix~\ref{AppQuad} shows, more information is optimally revealed as the highly responsive receiver becomes less responsive to the state (that is, as $\rho'_R(\theta)/\rho'_{FB}(\theta)$ decreases).}

As an aside, we briefly discuss how optimal relational communication depends on the monitoring structure. We have assumed that the decision is publicly observed, but the state is only observed by the sender. Alternatively, we might assume that the state is publicly observed at the end of each period. In this case, non-monotone message rules could be enforced by conditioning continuation play on the state. However, by the proof of Proposition~\ref{pr:sb}, the second-best unconstrained message rule turns out to be monotone, so Proposition \ref{pr:sb} and Corollary \ref{c:cs} would continue to hold. Yet alternatively, we might assume that the receiver never observes the state and the sender never observes the decision.\footnote{\cite{kuvalekar2019goodwill} study a repeated communication game with such a monitoring structure but without transfers. \cite{horner2015truthful} consider a much richer setting with communication and general monitoring structures, but focus on the case of sufficiently patient players.} Then the receiver would always take her preferred decision; thus complete pooling would be optimal if the receiver is highly responsive and full separation would be optimal otherwise.

The results of this section extend to the case where the receiver is highly responsive and is upward-biased for some states and downward-biased for others.\footnote{This case is technically challenging because, unlike the case where the receiver is downward-biased, the optimal unconstrained message rule is not necessarily monotone, so existing results from the literature on Bayesian persuasion no longer apply. A working-paper version of this paper (\citealt{kolotilin2019relationalwp}) derives new results on monotone persuasion (Section 4) and applies them to characterize optimal equilibria in this case (Section 5.3).  The working paper also derives a closed-form solution for the case where the state is uniformly distributed (Section 5.4).} Both high and low states may be high-conflict, with intermediate states being low-conflict. In this case, optimal communication takes one of three forms (Figure \ref{Fig:SwitchBias}). In all forms, overpooling occurs: high-conflict states are pooled with adjacent low-conflict states. Some intermediate interval of low-conflict states may be separated. As the players become more impatient, the high-conflict-state intervals expand, as do the corresponding pooling intervals. Eventually, the interval of separated states vanishes, so that the state space is partitioned into two pooling intervals. The two-pooling-interval structure remains optimal as the players become yet more impatient -- until, at some sufficiently small discount factor, the two pooling intervals suddenly coalesce into one, and complete pooling becomes optimal.

	\begin{figure}[t!]
	\centering 
		\includegraphics[scale=0.75]{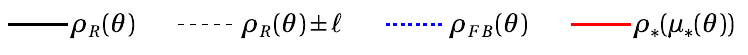}\\ \vspace{-1em}
	\subfloat[Separation in the middle]
	{\label{sub:Switch1}
	\includegraphics[scale=0.57]
	{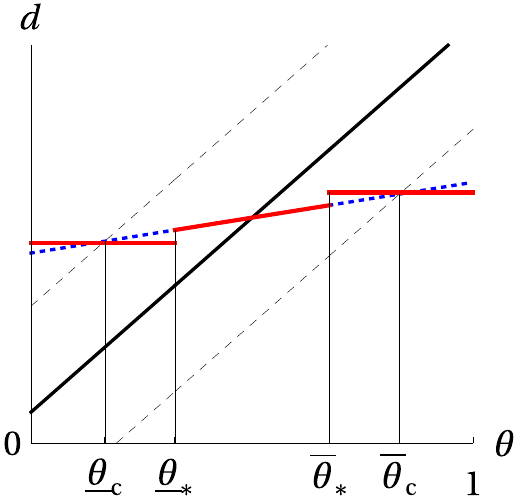}}
	\quad 
	\subfloat[Two intervals of pooling]
	{\label{sub:Switch2}
	\includegraphics[scale=0.57]
	{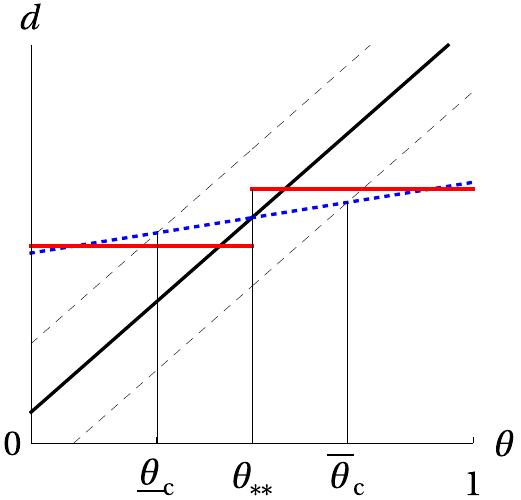}} 
	\quad 
	\subfloat[Complete pooling]
	{\label{sub:Switch3}
	\includegraphics[scale=0.57]
	{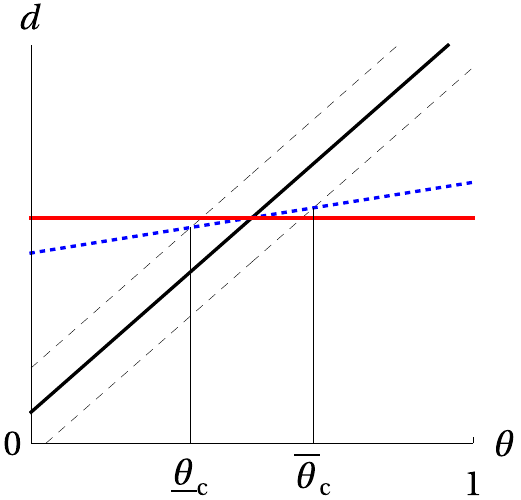}}\\
    \medskip
	\begin{minipage}{.95\textwidth}
	{\tablenote Optimal decision rules and decision outcomes given a highly-responsive receiver and two intervals of high-conflict states, $[0,\underline{\theta}_{\rm{x}})$ and $(\overline{\theta}_{\rm{x}},1]$. Figure \ref{sub:Switch1} shows the case of relatively high $\delta$, with two pooling intervals $[0,\underline{\theta}_{*})$ and $(\overline{\theta}_{*},1]$ and separation of the middle states $\theta \in (\underline{\theta}_{*},\overline{\theta}_{*})$. Figure \ref{sub:Switch2} shows the intermediate-$\delta$ case, with two adjacent pooling intervals $[0,\theta_{**})$ and $[\theta_{**},1]$. Figure \ref{sub:Switch3} shows the low-$\delta$ case with complete pooling. Overpooling occurs in all cases. \par}
	\end{minipage}	
	\caption{Optimal pooling with a highly-responsive receiver who has changing-sign bias
	} 
	\label{Fig:SwitchBias}
\end{figure}

\section{Punishment}\label{s:pun}

In this section, we characterize each player's worst equilibrium: the communication and decision-making outcomes that are used, as part of an optimal equilibrium, to punish the player as harshly as possible following a deviation. Unsurprisingly, the receiver is punished by complete pooling of information, so that she takes an uninformed decision.

We show in the general setting of Assumption \ref{CS} that the sender is punished by the highest or lowest incentive compatible decision. The sender's worst message rule may combine pooling and separation, which hurt the sender in distinct ways. Pooling misadapts decisions to the state, while separation extracts signaling transfers from the sender.

 The worst equilibrium payoffs $\underline{v}_R$ and $\underline{v}_S$ can be supported by single-period punishment strategy profiles in which a deviator is punished as harshly as possible but only for a single period.\footnote{\cite{baker1994subjective,Baker:2002ei} restrict attention to trigger-strategy equilibria where off-path punishments correspond to some static equilibria of the stage game; such trigger strategies are suboptimal in our setting because static equilibria are not the harshest possible punishments. Alternatively, \cite{levin2003relational} specifies exogenous outside option payoffs. Our results continue to hold in these alternative approaches, with the worst equilibrium payoffs replaced by either static equilibrium payoffs or outside option payoffs.}
  A \emph{single-period punishment} strategy profile is described by: normal as well as penal ex-ante transfers, $\underline{\tau}_0$ as well as  $\underline{\tau}_R$ and $\underline{\tau}_S$; normal as well as penal message rules, $\underline{\mu}_0$ as well as $\underline{\mu}_R$ and $\underline{\mu}_S$; normal as well as penal decision rules, $ \underline{\rho}_0$ as well as  $\underline{\rho}_R$ and $\underline{\rho}_S$.
 
 Play proceeds as follows. The ex-ante transfer is $\underline{\tau}_i$ if player $i\in \{{0}, R, S\}$ deviated last in the previous period, where $i=0$ denotes that no player deviated. The message rule, interim transfer rule, decision rule, and punishment message are $\underline{\mu}_j$, $\underline{t} _j$, $\underline{\rho}_j $, and $\underline{m}_j ^P$ if player $j\in \{0, R, S\}$ deviated from the ex-ante transfer in this period, where $\underline{t} _j$ and $\underline{m}_j ^P$ are defined by (\ref{eq:h}) and (\ref{eq:ts}) given $\underline{\mu}_j $ and $\underline{\rho}_j $. The punishment decision is $\underline{d}_j^P= \underline{\rho}_j (\underline{m}_j^P)$ if the sender deviated to some $(\hat{m},\hat{ t})\notin (\underline{\mu}_j, \underline{t}_j)([0,1])$ in this period. Ex-post transfers are always zero.

\begin{proposition}\label{prop:spp}
There exists an optimal equilibrium in single-period punishment strategies where:
\begin{enumerate}
\item The on-path rules are $\underline{\mu}_0=\mu_*$ and $\underline{\rho}_0=\rho_*$,\; so $\overline{v}= \mathbb E [u_*(\mu_*(\theta))]$;
\item The receiver's penal rules are $\underline{\mu}_R=[0,1]$ and $\underline{\rho}_R=\rho_R$,\; so $\underline{v}_{R} =u_R(\rho_R([0,1]),[0,1])$;
\item The sender's penal rules $\underline{\mu}_S $ and $\underline{\rho}_S$ solve
\begin{equation}\label{eq:worst}
\begin{gathered}
\underline{v}_S=\min_{ \mu,\rho,\theta^P} \left\{ u_S(\rho(m^P),\theta^P) + \mathbb E \left[ \int_{\theta^P}^\theta\frac{\partial u_S}{\partial \theta}({\rho}(\mu(\tilde\theta)),\tilde\theta)d\tilde \theta \right]\right\}\\
\text{subject to }\rho(\mu(\theta))\text{ is non-decreasing in }\theta,\\
\rho(m)
\begin{cases}
=\rho_-(m), &\text{if }m>m^P,\\
\in \{\rho_-(m),\rho_+(m)\}, &\text{if }m=m^P,\\ 
=\rho_+(m), &\text{if }m<m^P,
\end{cases}
\end{gathered}
\end{equation}
where $m_P=\mu(\theta_P)$ and $[\rho_-(m),\rho_+(m)]=\{d:w(d,m)\leq L(\overline v)\}$ for all $m \in \mu([0,1])$.
\end{enumerate}
\end{proposition}

Proposition \ref{prop:spp} specifies optimal punishments for the receiver and the sender: a deviator is punished as harshly as possible for a single period, and then optimal play resumes. 
The deviator's worst equilibrium payoff equals his or her payoff in the punishment period.
Following a deviation from $\tau$ by the receiver, the message rule is completely uninformative and no transfers are made. Following a deviation from $\tau$ by the sender, the receiver makes either the highest or lowest incentive compatible decision, and the message and interim transfer rules are chosen to minimize the sender's expected payoff. 

\medskip
\textit{The Uniform-Quadratic Example.}
Suppose that the players' payoffs are quadratic and the receiver is downward-biased, as in Section \ref{sec:quadratic}. Suppose further that the state is uniformly distributed. As shown in Proposition \ref{prop:worst} in Appendix \ref{a:pun}, the (upward-biased) sender's penal decision rule is the lowest incentive compatible decision rule, $\underline{\rho}_S(m)=\rho_R(m)-\ell$ for all $m$. Moreover, the sender's penal message rule either pools some interval of low states $[0,\underline\theta_S)$ or separates all states. As the players become more patient, signaling transfers become relatively more effective as a punishment; thus the optimal penal pooling interval shrinks, and full separation eventually becomes optimal.

\section{Public Information}\label{Sep}

In this section, we show that increasing transparency by adding public information generally worsens the relationship. The availability of transfers as a signaling device implies that better informed decision-making can always be achieved without tightening incentive constraints, so transparency adds no informational benefits for the relationship, but tightens incentive constraints in two ways. First, it improves both players' worst possible equilibrium payoffs, and thus limits the severity of off-path punishments. Second, it prevents information pooling, and thus limits the ability to discipline decision-making in high-conflict states.

We augment our model so that at the start of each period, the receiver observes a realization of a state-dependent signal. We consider the general payoffs of Assumption~\ref{CS}, rather than the quadratic payoffs of Assumption~\ref{quad}.
Just as with message rules, we assume that the signal rule $\psi(\theta)$ is deterministic and (without loss) identify each signal realization $s$ with the set of states that induce it, $s=\{\theta:\psi(\theta)=s\} $. 
We also assume that the signal rule $\psi$ is monotone in the sense that each $s\in\psi([0,1])$ is a convex set. 

Since the signal and message rules are deterministic, we can restrict attention to message rules that are \emph{refinements} of the signal rule in that for each realization $s$ of $\psi$ there exists a realization $m$ of $\mu$ such that $m\subset s$. In particular, this restriction allows us to consider decision rules $\rho$ that depend on the message $m$ but not the signal realization $s$, because $m$ incorporates all information contained in $s$.

The set of optimal equilibrium payoffs under signal $\psi$ can be computed by applying Proposition~\ref{prop:spp} separately to each realization of signal $\psi$. In particular, the optimal and penal message and decision rules are defined for each signal realization $s\in \psi ([0,1]) $ as follows: $\rho_*$ is given by (\ref{dstar}) and $\mu_*$ solves (\ref{mustar}) given that the set of states is $s$ rather than $[0,1]$;  $\underline \rho_R=\rho_R$ and $\underline \mu _R=\psi$; and $\underline\rho_S$ and $\underline \mu_S$ solve (\ref{eq:worst}) given that the set of states is $s$. 

We say that $\psi$ is \emph{more informative} than $\hat{\psi}$ if $\psi$ is a refinement of $\hat{\psi}$.
  For monotone signal rules, this notion coincides with the informativeness criterion of \cite{blackwell1953equivalent}. 

\begin{proposition}\label{prop:L_info}
Suppose that signal $\psi$ is more informative than signal $\hat{\psi}$. If the sender's penal decision outcome $\underline{\rho}_S(\underline{\mu}_S (\theta))$ is non-decreasing in $\theta$ under $\psi$, then the best joint payoff is lower under $\psi$ than under $\hat\psi$.%
\end{proposition}

To build intuition for Proposition \ref{prop:L_info}, we start with the myopic case. Consider moving from a fully informative signal ($\psi_f(\theta)=\theta$) to a completely uninformative signal ($\psi_u(\theta)=[0,1]$). We will argue that the worst equilibrium payoffs $\underline{v}_S$ and $\underline{v}_R$ strictly decrease and the best joint payoff $\overline{v}$ weakly increases.

The receiver's worst equilibrium payoff $\underline{v}_R$ is lower under $\psi_u$ than $\psi_f$. In the receiver's worst equilibrium, the receiver always chooses her preferred decision $\rho_R(\psi(\theta))$ given the signal $\psi$ and always receives zero transfers. Public information improves the receiver's decision-making and thus her worst equilibrium payoff.

The sender's worst equilibrium payoff $\underline{v}_S$ is lower under $\psi_u$ than $\psi_f$. The basic idea is that any equilibrium decision outcome implemented under $\psi_f$ (and thus a fully informed receiver) 
can also be implemented under $\psi_u$ by inducing the sender to fully reveal the state to the receiver. The sender's payoff $\underline{v}_S$ is strictly smaller under $\psi_u$ because inducing full separation requires the sender to make positive interim transfers to the receiver.

The best joint payoff $\overline{v}$ is weakly higher under $\psi_u$ than $\psi_f$, again because any equilibrium under $\psi_f$ can be implemented under $\psi_u$. In fact, the best joint payoff may be strictly higher under $\psi_u$ than $\psi_f$. For example, under $\psi_u$, the joint payoff is maximized under complete pooling of the states if the payoffs are quadratic and the receiver is highly responsive (see Section \ref{sec:QuadOpt}). Such pooling, however, is precluded under $\psi_f$ (and thus a fully informed receiver).

In the non-myopic case, these effects are preserved, and are further amplified by the shadow of the future. Moving from $\psi_f$ to $\psi_u$ increases the net discounted surplus $L(\overline v)$ (which increases with $\overline{v}$ and decreases with $\underline{v}_S$ and $\underline{v}_R$). This in turn relaxes constraints on decision-making and increases the best joint payoff.%

A subtle issue arises in the non-myopic case. An implementable decision outcome $\rho(\mu(\theta))$ under an informative signal $\psi$ must be non-decreasing in $\theta$ on each signal realization, but it is not required to be non-decreasing across signal realizations. Thus, there are non-monotone decision outcomes implementable under $\psi$ but not under a less informative signal $\hat\psi$. Nevertheless, to show that the best joint payoff is higher under $\hat \psi$ than under $\psi$, it is sufficient for the second-best and penal decision outcomes under $\psi$ to be monotone, and thus implementable under $\hat \psi$. By Proposition~\ref{prop:spp}, the second-best and receiver's penal decision outcomes are always monotone, but the sender's penal decision outcome can be non-monotone. Hence, in Proposition \ref{prop:L_info}, we require $\underline{\rho}_S(\underline{\mu}_S (\theta))$ to be non-decreasing in $\theta$ under $\psi$. This requirement is satisfied, for example, in the uniform-quadratic example of Section \ref{s:pun}, where $\underline{\rho}_S$ is the lowest incentive compatible decision rule.\footnote{Moreover, this requirement holds if $\delta=0$, or if $\psi(\theta)=\theta$ for all $\theta\in [0,1]$ and $\rho_R(\theta)$ does not cross $\rho_{FB}(\theta)$ from above. In general, however, $\underline{\rho}_S(\underline{\mu}_S(\theta))$ can be non-monotone. For example, suppose that the players' payoffs are quadratic, and $\rho_S(\theta)$ crosses $\rho_R(\theta)$ from below at $\theta_0 \in (0,1)$. Then $\underline \rho_S(\underline\mu_S(\theta))$ is not monotone under fully informative $\psi_f$. Indeed, for states $\theta<\theta_0$, the sender is downward-biased, so $\underline \rho_S(\underline \mu_S(\theta))=\rho_R(\theta)+\ell$. But, for states $\theta>\theta_0$, the sender is upward-biased, so $\underline \rho_S(\underline \mu_S(\theta))=\rho_R(\theta)-\ell$. Thus, $\underline \rho_S(\underline \mu_S(\theta))$ jumps down at $\theta_0$.}

The result that public information hurts the relationship relates to various papers that study the social value of public information. \cite{hirshleifer1971private} argues that welfare may be decreasing in the amount of public information available to agents. \cite{bergemann2015bayes} clarifies this point: making more information available to an agent may, by increasing the set of incentive constraints she faces, shrink the set of equilibrium outcomes.\footnote{\cite{cremer1995arm}, \cite{kloosterman2015public}, and \cite{Fong:2015vc} discuss other settings where public information may be detrimental.} This relates to the logic of our model, where the availability of public information makes it impossible to pool incentive constraints across states, and thus worsens incentive provision within the relationship. Public information in our model also improves the worst possible equilibrium payoffs for both players; this decreases the surplus and thus worsens intertemporal incentives.\footnote{This point relates to an insight from \cite{baker1994subjective}. There, objective performance measures, rather than transparency, improve the players' outside options and make cooperation within the relationship more difficult to sustain.}

\section{Conclusion}

In our model, incomplete information transmission does not reflect a failure to motivate communication, but instead is an instrument for managing decision-making.
 This finding relies on the capacity of voluntary transfers to credibly support any monotone message rule at no welfare cost. It suggests that when modeling strategic communication in applied settings, it is crucial to understand whether monetary or non-monetary transfers (such as wages or favours) are available, because our implications differ significantly from those of the standard literature on strategic communication without transfers. In fact, one interpretation of our model is that voluntary transfers endogenously endow the privately-informed sender with the ability to commit to any monotone message rule, even with impatient players. Such commitment is the premise of the literature on Bayesian persuasion (\citealt{kamenica2011bayesian}). So, our analysis extends the applicability of the Bayesian persuasion framework to settings without commitment but with transfers.

Our model is remarkably tractable and thus allows for a thorough treatment of repeated interactions. This analysis produces a rich and intuitive set of results. In particular, incomplete information transmission is implemented only for high-conflict states, and only if the receiver's decision-making is too responsive to information. One implication is that with constant bias, pooling does not occur. In contrast, in the standard constant-bias cheap-talk game (\citealt{crawford1982strategic}), information transmission is always incomplete, and this is generally exacerbated in high (low) states if the sender is upward- (downward-) biased.

Our result that adding public information worsens optimal equilibria highlights the benefits of an `arms-length' approach where information and control are separated. For instance, giving the decision right to the sender may worsen the relationship by increasing the players' worst equilibrium payoffs and thus reducing the severity of punishments. Relatedly, introducing mediators who control the flow of information from the sender to the receiver cannot improve the relationship. This is because it is optimal to give the sender as much control over the release of information as possible.  

We hope that future work will use our tractable framework to study other challenging problems in strategic communication. For example, one might examine the case with multiple senders and receivers, possibly connected by a communication network. Another promising avenue would be to allow for costly information acquisition.

\begin{appendices}

\section{Stationarity}

This appendix specifies necessary and sufficient conditions for equilibrium, and proves Lemma \ref{lem:stationary}.

To show that the set of equilibrium payoffs is compact, restrict decisions and transfers to compact sets $d\in[-\overline{d},\overline{d}]$ and $\tau,t,T\in [-\overline{t},\overline{t}]$. Under this restriction and under Assumption~\ref{CS}, it can be shown that the set of equilibrium payoffs is compact (see, for example, \citealt{mailath2006repeated}). Now, observe that this restriction is without loss of generality if the bounds $\overline{d}$ and $\overline{t}$ are chosen to be large enough that (in any equilibrium) decisions and transfers are interior. Indeed, under Assumption~\ref{CS}, we can show that such bounds exist.

Let $V$ be the set of equilibrium payoffs. We now show that the set $\overline V$ of optimal equilibrium payoffs is  given by \eqref{simplex}.
Consider an optimal equilibrium payoff vector $\left( v_{S}^{\ast
},v_{R}^{\ast }\right) $ that maximizes the joint payoff, so that $v_{S}^{\ast }+v_{R}^{\ast }=\overline{v}$, and let $\sigma_*$ be an equilibrium supporting $\left( v_{S}^{\ast
},v_{R}^{\ast }\right) $.
Let $\left( v_{S},v_{R}\right) $ be any point in $\overline V$. 
Notice that we can modify $\sigma_*$ to produce $\left( v_{S},v_{R}\right) $ by changing only the ex-ante transfer in the first period from $\tau^*$ to $\tau=\tau^*+(v_{S}^{\ast }-v_{S})/(1-\delta)$. This modification affects the players' incentives only at the ex-ante round of the first period. Each player is willing to make the ex-ante transfer $\tau$ because $v_{S} \geq \underline{v}_{S}$ and $v_{R} \geq \underline{v}_{R}$ by definition of $\overline V$. Thus, the modified strategy profile is an equilibrium. Moreover, it is an optimal equilibrium, because $\overline v $ is the maximum equilibrium joint payoff. Finally, if $v_S<\underline v_S$ or $v_R<\underline v_R$, then $(v_S,v_R)$ cannot be supported in equilibrium. We conclude that $\overline V$ is the set of optimal equilibrium payoffs.

A message rule $\mu(\theta)$, a decision rule $\rho(m)$, transfer rules $\tau$, $t(m)$, $T(m)$, continuation payoff function $v_{i}(m)$, for each $i\in \{S,R\}$, and punishment decision $d^P$ and message $m^P$ constitute an equilibrium if and only if the following six conditions hold (see, for example, \citealt{mailath2006repeated}):

\begin{enumerate}
\item[C1.] Each player is willing to make ex-ante transfer $\tau$:%
\begin{eqnarray*}
v_{S} = (1-\delta )[-\tau +\mathbb{E}[u_{S}(\rho(\mu(\theta
)),\theta )-t(\mu(\theta ))-T(\mu(\theta) )]]+\delta \mathbb{E}%
[v_{S}(\mu(\theta) )]\geq \underline{v}_{S};\\
v_{R} = (1-\delta )[\tau +\mathbb{E}[u_{R}(\rho(\mu(\theta
)),\theta )+t(\mu\left( \theta \right) )+T(\mu(\theta
))]]+\delta \mathbb{E}[v_{R}(\mu(\theta) )]\geq \underline{v}_{R}.
\end{eqnarray*}

\item[C2.] For each state $\theta$, the sender is willing to send  message $\mu(\theta)$ and to make interim transfer $t(\mu(\theta) )$.

\begin{enumerate}
\item There is no profitable deviation to another message -- interim-transfer pair\\ $\left(\mu(\hat{\theta}),t(\mu(\hat{\theta}))\right)$ that is observed on the equilibrium path: 
\begin{eqnarray*}
&&(1-\delta )[u_{S}(\rho(\mu(\theta )),\theta )-t(\mu(\theta
))-T(\mu(\theta) )]+\delta v_{S}(\mu(\theta) ) \\
& \geq &(1-\delta )[u_{S}(\rho(\mu(\hat{\theta})),\theta
)-t(\mu(\hat{\theta}))-T(\mu(\hat{\theta}) )]+\delta v_{S}(\mu(\hat \theta) )\text{~for all%
}~\theta ,\hat{\theta}\in [0,1].
\end{eqnarray*}

\item There is no profitable deviation to some pair $(\hat{m},\hat{t})$ that is never observed on the equilibrium path:%
\begin{eqnarray*}
&&(1-\delta )[u_{S}(\rho(\mu(\theta )),\theta )-t(\mu(\theta
))-T(\mu(\theta ))]+\delta v_{S}(\mu(\theta ))\\
&\geq &(1-\delta
)u_{S}(d^P,\theta )+\delta \underline{v}_{S}~\text{for all}~\theta\in [0,1].
\end{eqnarray*}
Here, we specify that following any such deviation, the receiver chooses punishment decision $d^P$.
\end{enumerate}

\item[C3.] The receiver is willing to make interim transfer $t(m)$:
\begin{eqnarray*}
&& (1-\delta )[u_{R}(\rho(m),m
)+t(m)+T(m )]+\delta v_{R}(m) \\
&\geq & (1-\delta )%
u_{R}(\rho_R(m),m)+\delta \underline{v}_{R}~\text{for
all}~m\in \mu([0,1]).
\end{eqnarray*}
\item[C4.] The receiver is willing to choose decision $\rho(m)$ on-path and $d^P$ off-path.
\begin{enumerate}
\item After an on-path message -- interim-transfer pair, the receiver is willing to choose decision $\rho(m)$:
\begin{eqnarray*}
&& (1-\delta )[u_{R}(\rho(m),m )+T(m )]+\delta 
v_{R}(m)\\ 
&\geq& (1-\delta )u_{R}(\rho_R(m),m )+\delta \underline{v}_{R}~\text{for all}~m\in \mu([0,1]).
\end{eqnarray*}
\item After an off-path message -- interim-transfer pair, the receiver is willing to choose decision $d^P$:
\begin{eqnarray*}
(1-\delta )u_{R}(d^P,m^P )+\delta (\overline{v}-\underline{v}_S) &\geq& (1-\delta)u_{R}( \rho_R(m^P) ,m^P) + \delta \underline{v}_{R}.
\end{eqnarray*}
Here, we specify that following any deviation by the sender, the receiver believes that the state is in $m^P\subset [0,1]$.
\end{enumerate}

\item[C5.] Each player is willing to make ex-post transfer $T(m)$: 
\begin{align*}
-(1-\delta )T(m )+\delta v_{S}(m)&\geq \delta \underline{v}_{S}~\text{for all}~m\in \mu([0,1]);\\
(1-\delta )T(m )+\delta v_{R}(m )&\geq \delta \underline{v}_{R}~\text{for all}~m\in \mu([0,1]).
\end{align*}

\item[C6.] The continuation payoffs can be supported in equilibrium:%
\begin{align*}
\left( v_{S}(m),v_{R}(m)\right) &\in V~\text{for all}~m\in \mu([0,1]).%
\end{align*}
\end{enumerate}

\begin{proof}[Proof of Lemma \ref{lem:stationary}]
We have already shown that the set of optimal equilibrium payoffs is $\overline V$ given by (\ref{simplex}). In any optimal equilibrium, continuation is optimal: $v_{S}(m)+v_{R}\left( m \right) =\overline{v}$ for all $m\in \mu([0,1]) $.
Otherwise, one could  increase $v_S(m)$ and $v_{R}(m)$ without violating Condition C6, thereby relaxing the constraints of Conditions C1--C5 and increasing joint payoff $v_S + v_R$.

An optimal equilibrium $\sigma$ with zero first-period ex-ante transfers clearly exists. Let $(v_S,v_R)\in \overline V$ be the payoff vector under $\sigma$. We will modify $\sigma$ to construct an optimal stationary equilibrium with the same payoff vector. Let $\mu(\theta)$, $\rho(m)$, $t(m)$, $T(m)$, and $v_i(m)$, for each $i\in \{S,R\}$, be the message rule, decision rule, transfer rules, and continuation payoff function in the first period on the equilibrium path of $\sigma$. Define $T^{\ast }\left( m \right) $ by
\begin{align*}
-\left( 1-\delta \right) T^{\ast }\left( m \right) +\delta
v_{S} &= -(1-\delta )T(m )+\delta v_{S}(m ).	
\end{align*}
Since $v_{S}+v_{R}=v_{S}\left( m \right) +v_{R}\left( m \right) =\overline{v}$ by optimality of $\sigma$, we also have 
\begin{align*}
\left( 1-\delta \right) T^{\ast }\left( m \right) +\delta
v_{R} &= (1-\delta )T(m )+\delta v_{R}(m ).	
\end{align*}

Consider the following stationary strategy profile $\sigma_ *$. On the equilibrium path,  $\mu(\theta )$, $\rho(m)$, $\tau=0$, $t(m)$,  and $T^*(m )$ are played in each period. Following any deviation, except for an undetectable deviation by the sender as in Condition C2(a), play proceeds according to $\sigma$. By construction, the sender's and receiver's expected payoffs under $\sigma_*$ are the same as under $\sigma$.

We now show that $\sigma_*$ constitutes an equilibrium. In each period the constraints of Conditions C1 -- C5 continue to hold under $\sigma_*$ because they are identical to the first-period constraints under $\sigma$. Condition C6 holds because $(v_S,v_R)$ belongs to $\overline V$ by supposition.

Finally, by modifying the first-period ex-ante transfer in $\sigma_*$ from $0$ to $\tau=(v_{S}-\hat{v}_{S})/(1-\delta)$, we can support any optimal equilibrium payoff vector $(\hat{v}_S,\hat{v}_R)\in \overline V$.
\end{proof}

\begin{lemma}\label{mondelta}
If $0\leq \delta < {\delta}'<1$, then the corresponding best joint payoffs satisfy $\overline v\leq {\overline v}'$, and the corresponding worst equilibrium payoffs satisfy ${\underline v}_S\geq {\underline v}_S'$ and ${\underline v}_R\geq {\underline v}_R'$. %
\end{lemma}
\begin{proof}
Given $\delta\in [0,1)$, consider a stationary optimal equilibrium $\sigma_*$ with zero ex-ante transfers. Let this equilibrium produce an equilibrium payoff vector $(v^*_S,v^*_R)$, with $v^*_S+v^*_R=\overline{v}$. We can support any optimal equilibrium payoff vector $(v_S,v_R)\in \overline V$ under $\delta$ by modifying the first-period ex-ante transfer in $\sigma_*$ from $0$ to $\tau=(v_{S}^{\ast }-v_{S})/(1-\delta)$. Notice that Conditions C1 -- C6 continue to hold under $\delta'\in(\delta,1)$, after replacing $\tau=(v_{S}^{\ast }-v_{S})/(1-\delta)$ with $\tau'=(v_{S}^{\ast }-v_{S})/(1-\delta')$, because
\begin{equation*}
	\frac{\delta'}{1-\delta'}(v_i^*-\underline{v}_i)\geq \frac{\delta}{1-\delta}(v_i^*-\underline{v}_i)\text{ for each }i\in \{S,R\}.
\end{equation*}
Therefore, the set $\overline V$ is self-generating under $\delta'$, which proves the lemma (see, for example, \citealt{mailath2006repeated}).
\end{proof}

\section{Equilibrium}\label{app:eq}

\begin{proof}[Proof of Proposition \ref{pr:impl}]
Consider a stationary equilibrium $\sigma$ that produces a joint payoff $v$. Let $\mu(\theta)$, $\rho(m)$, $\tau$, $t(m)$, and $T(m)$ be the message rule, decision rule, and transfer rules on the equilibrium path of $\sigma$. Define $U_S(\theta)$ as the one-period payoff of the sender if the state is $\theta$,
\begin{equation*}
	U_S(\theta)=u_{S}(\rho(\mu(\theta )),\theta )-p(\theta),
\end{equation*}
where $p(\theta)$ is the net one-period transfer of the sender if the state is $\theta$, 
\begin{equation*}
p(\theta)=\tau+t(\mu(\theta))+T(\mu(\theta) ).
\end{equation*}
Condition C2 (a) requires that
\begin{equation*}
U_S(\theta) \geq  u_{S}(\rho(\mu(\hat{\theta})),\theta )-p(\hat{\theta})\text{ for all $\theta,\hat{\theta}\in [0,1]$}. %
\end{equation*}
Since $\partial^2 u_S(d,\theta)/\partial d \partial \theta>0$ by Assumption~\ref{CS}, this inequality holds if and only if $\rho(\mu(\theta))$ is non-decreasing in $\theta$ and 
\begin{equation}
U_S(\theta)=U_S(0)+\int_0^\theta \frac{\partial u_S}{\partial \theta} (\rho(\mu(\tilde\theta)),\tilde\theta) d\tilde\theta\text{ for all }\theta\in [0,1], \label{envelope}	
\end{equation}
by Proposition 1 of \cite{rochet1987necessary} and Corollary 1 of \cite{milgrom2002envelope}.

Adding the constraint of Condition C4 (a) and the sender's constraint of Condition C5, and taking into account that $v_S+v_R =v$, gives (\ref{enforcement}).

Conversely, suppose that $\mu(\theta)$ and $\rho(m)$ are such that $\rho(\mu(\theta))$ is non-decreasing in $\theta$ and (\ref{enforcement}) holds. We construct transfer rules and punishment variables that satisfy Conditions C1 -- C6, and thus constitute a stationary equilibrium. We consider the case $\delta>0$; the case $\delta=0$ is simpler but slightly different. Let $T(m)=0$ and $\tau=\mathbb{E}[u_{S}(\rho(\mu(\theta)),\theta)-t(\mu(\theta))]-\underline{v}_S$. Moreover, let $t(m)$ and $m^P$ be defined by (\ref{eq:h}) and (\ref{eq:ts}), and let $d^P=\rho(m^P)$.  

Equation (\ref{eq:ts}) assumes that $t(m^P)=\inf_{m\in \mu([0,1])}t(m)$ for some message $m^P\in \mu([0,1])$. If this assumption does not hold, then we specify $m^P$ and $d^P$ as follows. By the Bolzano-Weierstrass theorem, there exists a sequence $\{m_k\}\in \mu([0,1])$ such that as $k\rightarrow \infty$, $t(m_k)\rightarrow \inf_{m\in \mu([0,1])} t(m)$, $\theta(m_k)\rightarrow \theta_\star$, and $\rho(m_k)\rightarrow d_\star$ for some $\theta(m_k)\in m_k$, $\theta_\star\in [0,1]$, and $d_\star\in \mathbb{R}$. Set $m^P=\theta_\star$ and $d^P=d_\star$. Since $u_R(d,\theta)$ is continuous, and (\ref{enforcement}) holds for all $(m_k,\rho(m_k))$, it also holds for $(m^P,d^P)$.

Notice that  the left-hand side of (\ref{enforcement}) is non-negative, so $v\geq \underline{v}_S+\underline{v}_R$.
Condition C6 holds because the continuation payoffs are $v_S=\underline{v}_S$ and $v_R=v-\underline{v}_S$. Condition C5 holds because Condition C6 holds and $T(m)=0$. The sender's constraint of Condition C1 holds with equality. The receiver's constraint of Condition C1 holds because it can be simplified to $v\geq \underline{v}_S +\underline{v}_R$. Condition C2 (a) holds because $\rho(\mu(\theta))$ is non-decreasing in $\theta$ and (\ref{envelope}) holds. Condition C2 (b) holds because by deviating to a message-transfer pair $(\hat{m},\hat{t})$ that is not observed on the equilibrium path, the sender induces $d^P=\rho\left(m^P\right)$, which he can induce more cheaply on the equilibrium path with message $m^P$ and zero interim transfer $t\left( m^P\right) =0$. This argument assumes that there exists $m^P$ such that $t(m^P)=\inf_{m\in \mu([0,1])}t_S(m)$. Condition C2 (b) still holds even if such $m^P$ does not exist. This is because Condition C2 (a) holds for each $\hat{\theta}=\theta(m_k)$, and thus in the limit $k \to \infty$. But in this limit, Condition C2 (a) coincides with Condition C2 (b). Condition C4 is a restatement of (\ref{enforcement}). Note that, as for Condition C2 (b), a limiting argument needs to be made for Condition C4 (b) if $\inf_{m\in \mu([0,1])}t(m)$ is not attained by any $m^P$. Condition C3 holds because Condition C4 holds and $t(m)$ is non-negative.
\end{proof}

\begin{proof}[Proof of Proposition~\ref{pr:opt}]
By Lemma~\ref{lem:stationary} and Proposition~\ref{pr:impl}, in an optimal equilibrium, the decision and message rules solve
\begin{gather}
	\overline{v}=\max_{\mu,\rho}\mathbb{E}[u(\rho(\mu(\theta)),\theta)] \label{vbar}\\
	\text{subject to } \rho(\mu(\theta))\text{ is non-decreasing in }\theta, \label{vbar_a}\\
	w(\rho(m),m)\leq L(\overline{v})\text{ for all }m\in \mu([0,1]).  \label{vbar_b}
\end{gather}

Without loss of generality, we can restrict attention to monotone message rules. The argument is similar to the revelation principle. To this end, consider any $\mu$ and $\rho$ that satisfy (\ref{vbar_a}) and (\ref{vbar_b}). Define new rules $\tilde\mu$ and $\tilde \rho$ as $\tilde\mu(\tilde\theta)=\{\theta:\rho(\mu(\theta))= \rho( \mu(\tilde\theta))\}$ for all $\tilde \theta \in [0,1]$ and $\tilde \rho(\tilde m)=\rho(\mu(\tilde\theta(\tilde m)))$ for all $\tilde m\in \tilde \mu ([0,1])$, where $\tilde\theta(\tilde m)$ is an arbitrary state $\tilde\theta\in \tilde m$. It is easy to see that $\tilde \rho  (\tilde m)$ is independent of the choice of a representative state $\tilde\theta\in \tilde m$ and that $\tilde \rho(\tilde \mu (\theta))=\rho(\mu (\theta))$ for all $\theta\in [0,1]$. Since $\rho(\mu (\theta)) $ is non-decreasing in $\theta$ by (\ref{vbar_a}), $\tilde \rho(\tilde \mu (\theta)) $ is also non-decreasing in $\theta$ and $\tilde \mu $ is monotone. Moreover, since each set $\tilde m\in \tilde \mu([0,1])$ is the union of some disjoint sets $m\in \mu ([0,1])$ and the constraint (\ref{vbar_b}) holds for $\rho(m)$ for each $m\in \mu ([0,1])$, the constraint (\ref{vbar_b}) also holds for $\tilde \rho(\tilde m)$ for each $\tilde m\in \tilde \mu([0,1])$.

Consider a relaxed problem
\begin{gather*}
	\overline{v}=\max_{\mu,\rho}\mathbb{E}[u(\rho(\mu(\theta)),\theta)]\\
	\text{subject to } \mu\text{ is monotone}, \nonumber \\
	w(\rho(m),m)\leq L(\overline{v})\text{ for all }m\in \mu([0,1]). \label{relcon}
\end{gather*}
We can solve this relaxed problem in two steps. First, for a given monotone message rule $\mu$, the optimal decision rule specifies that the second-best decision is taken for each on-path message $m \in \mu([0,1])$. Second, given this optimal decision rule, the optimal message rule is clearly $\mu_*$ defined by (\ref{mustar}). To prove that the solution to the relaxed problem also solves the original problem (\ref{vbar}), it suffices to show that $\rho_*(m)$ is non-decreasing in $m$.

We first rewrite the constraint of the problem (\ref{dstar}) as $d\in D(m)$ where $D(m)$ is non-decreasing in $m$ in the strong set order. Since $u_R(d,\theta)$ is strictly concave in $d$ and has a unique maximum, $w(d,m)$ is strictly convex in $d$ and has a unique minimum. Taking into account that $w(\rho_R(m),m)=0$ and $L(\overline{v})\geq 0$, we have that the set of decisions $d$ that satisfy the constraint of the problem (\ref{dstar}) is a nonempty closed convex set and thus can be written as $D(m)=[\rho_-(m),\rho_+(m)]$, where $\rho_-(m)$ and $\rho_+(m)$ satisfy the constraint with equality. Moreover, since $u_R(d,\theta)$ is concave in $d$ and is supermodular, $w(d,m)$ is non-increasing in $d$ and non-decreasing in $m$ if $d<\rho_R(m)$, and $w(d,m)$ is non-decreasing in $d$ and non-increasing in $m$ if $d>\rho_R(m)$. This implies that $\rho_-(m)$ and $\rho_+(m)$ are non-decreasing in $m$, and thus $D(m)$ is non-decreasing in $m$. Taking into account that $u(d,m)$ is strictly concave and has increasing differences,   $\rho_*(m)=\arg\max_{d\in D(m)}u(d,m)$ is non-decreasing in $m$, as follows, for example, from Theorem 4 of \cite{milgrom1994monotone}.
\end{proof}

\begin{proof}[Proof of Proposition \ref{prop:localpool}]
Consider an interval $m=(\theta_1,\theta_2)$. Let all expectations condition on $m$, that is, $\mathbb{E}[x]=\mathbb{E}[x|m]$.
Let the expected state given $m$ be $\hat{\theta}=\mathbb{E}\left[\theta \right]$. Denote $\hat\rho_* = \rho_*(\hat\theta)$, $\hat\rho_R = \rho_R(\hat\theta)$, and $(\Delta\theta)^2 = \mathbb{E}\left[(\theta-\hat\theta)^2 \right]$. In our calculations, we fix $\hat\theta$ and consider small $m$: $\theta_1\uparrow \hat \theta$ and $\theta_2 \downarrow \hat{\theta}$. 

\noindent To start, we approximate $\mathbb{E}\left[  \rho_*(\theta)-\hat\rho_* \right]$. Suppose $\hat\theta$ is high-conflict, so the receiver's incentive constraint for singleton message $\hat\theta$ is strictly binding (and thus is also strictly binding for nearby states and intervals),
\begin{gather}
u_R(\rho_R(\theta),\theta)-u_R(\rho_*(\theta),\theta) = L(\overline{v}) \text{ for $\theta \in m$,}  \label{eqn:rhostar_nonmyopic} \\
\mathbb{E}[u_R(\rho_R(m),\theta)-u_R(\rho_*(m),\theta)] = L(\overline{v}). \nonumber
\end{gather}
Differentiating (\ref{eqn:rhostar_nonmyopic}) and noting that ${u_R}'_d(\rho_R(\theta),\theta)=0$, ${u_R}''_{d\theta}(d,\theta)>0$, and $\rho_*(\theta) \neq \rho_R(\theta)$:
\begin{gather*}
{u_R}'_d(\rho_R(\theta),\theta) \rho'_R(\theta) + {u_R}'_\theta(\rho_R(\theta),\theta)-{u_R}'_d(\rho_*(\theta),\theta)\rho'_*(\theta) - {u_R}'_\theta(\rho_*(\theta),\theta) = 0 \text{, so} \nonumber\\
\rho'_*(\theta) = \frac{{u_R}'_\theta(\rho_R(\theta),\theta)-{u_R}'_\theta(\rho_*(\theta),\theta)}{{u_R}'_d(\rho_*(\theta),\theta)} \in (0,\infty). \label{rhoprimenotzero}
\end{gather*}
We can then write
\begin{gather*}
	\mathbb{E}\left[ \rho_*(\theta)-\hat\rho_*\right] = \mathbb{E}\left[ \rho_*'(\hat \theta)(\theta-\hat\theta)+\frac{1}{2}\rho_*''(\hat \theta)(\theta-\hat\theta)^2+o(\theta-\hat\theta)^2 \right] 
	= \frac{1}{2}\rho_*''(\hat \theta)(\Delta\theta)^2 +o(\Delta \theta)^2 .
\end{gather*}
\noindent Next, we approximate $\mathbb{E}\left[u(\rho_*(m),\theta)-u(\hat\rho_*,\theta )) \right]$. Starting with
\[
u_R(\hat{\rho}_R,\hat{\theta})-u_R(\hat\rho_*,\hat{\theta})
=  L(\overline{v}) = \mathbb{E}[u_R(\rho_R(m),\theta)-u_R(\rho_*(m),\theta)],\]
we rearrange to obtain
\[\mathbb{E}[u_R(\rho_R(m),\theta)-u_R(\hat{\rho}_R,\hat{\theta})] = \mathbb{E}[u_R(\rho_*(m),\theta)-u_R(\hat\rho_*,\hat\theta)]. \]
The left-hand side expands to become
\begin{gather*}
\mathbb{E}[u_R(\rho_R(m),\theta)-u_R(\hat{\rho}_R,\hat{\theta})]\\
=  \mathbb{E} \left[
\begin{matrix}
 {u_R}'_{d}(\hat{\rho}_R,\hat{\theta}) (\rho_R(m)-\hat\rho_R) + {u_R}'_{\theta}(\hat{\rho}_R,\hat{\theta}) (\theta - \hat{\theta})
 + {u_R}''_{d\theta}(\hat{\rho}_R,\hat{\theta}) (\rho_R(m)-\hat{\rho}_R)(\theta - \hat{\theta})\\
  + \frac{1}{2}{u_R}''_{dd}(\hat{\rho}_R,\hat{\theta}) (\rho_R(m)-\hat\rho_R)^2  + \frac{1}{2}{u_R}''_{\theta\theta}(\hat{\rho}_R,\hat{\theta}) (\theta - \hat{\theta})^2 
 +o((\rho_R(m)-\hat{\rho}_R)^2 + (\theta - \hat{\theta})^2)
\end{matrix}
\right]	\\
=
 \frac{1}{2} {u_R}''_{dd}(\hat{\rho}_R,\hat{\theta}) (\rho_R(m)-\hat{\rho}_R)^2 + \frac{1}{2}{u_R}''_{\theta\theta}(\hat{\rho}_R,\hat{\theta}) (\Delta \theta)^2  +o((\rho_R(m)-\hat{\rho}_R)^2 + (\Delta \theta)^2)\\
= \frac{1}{2}{u_R}''_{\theta\theta}(\hat{\rho}_R,\hat{\theta})(\Delta \theta)^2 +o(\Delta \theta)^2;
\end{gather*}
while the right-hand side expands to become
\begin{gather*}
\mathbb{E}[u_R(\rho_*(m),\theta)-u_R(\hat\rho_*,\hat\theta)] = \mathbb{E} \left[\begin{matrix}
	{u_R}'_d(\hat\rho_*,\hat\theta)(\rho_*(m)-\hat\rho_*)+ \frac{1}{2}{u_R}''_{\theta \theta}(\hat\rho_*,\hat\theta)(\theta-\hat\theta)^2\\
	+o(\theta-\hat\theta)^2+o(\rho_*(m)-\hat\rho_*)
\end{matrix} \right].
\end{gather*}
Equating both sides then yields
\begin{gather*}
\rho_*(m)-\hat\rho_*=\frac{{u_R}''_{\theta\theta}(\hat{\rho}_R,\hat{\theta})-{u_R}''_{\theta \theta}(\hat\rho_*,\hat\theta) }{2 \, {u_R}'_d(\hat\rho_*,\hat\theta)}(\Delta \theta)^2 +o(\Delta \theta)^2.
\end{gather*}
Substituting in this last expression,
\begin{gather*}
\mathbb{E}\left[u(\rho_*(m),\theta)-u(\hat\rho_*,\theta )) \right] = \mathbb{E}\left[u'_{d}(\hat\rho_*,\theta) (\rho_*(m)-\hat\rho_*)  +o(\rho_*(m)-\hat\rho_*) \right]\\
=u'_{d}(\hat\rho_*,\hat{\theta}) \frac{{u_R}''_{\theta\theta}(\hat{\rho}_R,\hat{\theta})-{u_R}''_{\theta \theta}(\hat\rho_*,\hat\theta) }{2 {u_R}'_d(\hat\rho_*,\hat\theta)} (\Delta \theta)^2  + o(\Delta \theta)^2.
\end{gather*}
It is better to pool than separate all states in interval $m$ if and only if
\begin{align*}
0 & >  \mathbb{E} \left[u(\rho_*(\theta),\theta)-u(\rho_*(m),\theta) \right]\\
&=\mathbb{E} \left[
\begin{matrix}
u''_{d\theta}(\hat\rho_*,\hat\theta)(\theta-\hat \theta)(\rho_*(\theta)-\hat\rho_*)+\frac{1}{2}u''_{dd}(\hat\rho_*,\hat\theta)(\rho_*(\theta)-\hat\rho_*)^2\\
 + u'_d(\hat\rho_*,\hat \theta)(\rho_*(\theta)-\hat\rho_*) + u(\hat\rho_*,\hat\theta)-u(\rho_*(m),\hat\theta) + o(\Delta\theta)^2
\end{matrix} \right]\\
&=\begin{Bmatrix}
u''_{d\theta}(\hat\rho_*,\hat\theta)\rho'_*(\hat \theta)+\frac{1}{2}u''_{dd}(\hat\rho_*,\hat\theta)(\rho'_*(\hat \theta))^2	\\
 + \frac{1}{2}u'_d(\hat\rho_*,\hat \theta)\left( \rho_*''(\hat \theta) + \frac{{u_R}''_{\theta \theta}(\hat\rho_*,\hat\theta)-{u_R}''_{\theta\theta}(\hat{\rho}_R,\hat{\theta}) }{{u_R}'_d(\hat\rho_*,\hat\theta)} \right)
\end{Bmatrix} (\Delta \theta)^2
   + o(\Delta\theta)^2;
\end{align*}	
so full disclosure can be improved in some neighbourhood of $\hat\theta$ if
\begin{align*}
0 & > u''_{d\theta}(\hat\rho_*,\hat\theta)\rho'_*(\hat\theta)+\frac{1}{2}u''_{dd}(\hat\rho_*,\hat\theta)(\rho'_*(\hat\theta))^2	 + \frac{1}{2}  u'_d(\hat\rho_*,\hat\theta) \left( \rho_*''(\hat\theta) +\frac{{u_R}''_{\theta \theta}(\hat\rho_*,\hat\theta)-{u_R}''_{\theta\theta}(\hat\rho_R,\hat\theta)}{{u_R}'_d(\hat\rho_*,\hat\theta)} \right) \\
& = \frac{1}{2}{\rho'_*(\hat\theta)u''_{dd}(\rho_*( \hat\theta),\hat\theta)}\left(\rho'_*(\hat\theta)-2 \tilde \rho_{FB}'(\hat\theta)  \right)	+ {u'_d(\rho_*( \hat\theta),\hat\theta)}\left( \frac{\rho_*''(\hat\theta)}{2}  -  \frac{ \rho_*'(\hat\theta)\varphi''(\hat\theta)}{2 \varphi'(\hat \theta)}  \right),	
\end{align*}
where the equality can be verified by substituting $\tilde\rho_{FB}'(\theta)$, $\rho_*'(\hat\theta)$, $\varphi'(\hat\theta)$, and $\varphi''(\hat\theta)$.
Dividing this expression by $\rho_*'(\hat\theta) > 0$, we obtain (\ref{eqn:rhoprime}). 

Two final observations: first, if $\rho_*(\hat\theta)= \rho_{FB}(\hat\theta)$, then $\rho'_*(\hat\theta)=\rho'_{FB}(\hat\theta)=\tilde{\rho}'_{FB}(\hat\theta)>0$ and $u'_d(\rho_*( \hat\theta),\hat\theta)=0$; in which case (\ref{eqn:rhoprime}) must fail. Second, in Footnote \ref{f:locfb}, we remark that $\tilde\rho_{FB}'(\theta)$ is the slope of the 'locally first-best' decision rule.
To justify this interpretation, we approximate $\mathbb E[u(\hat \rho_*+a(\theta-\hat\theta),\theta)]$ as follows
\[
\mathbb E[u(\hat \rho_*+a(\theta-\hat\theta),\theta)]=u(\hat \rho_*,\hat\theta) +\left(\frac{1}{2}u_{dd}''(\hat\rho_*,\hat\theta)a^2 +u_{d\theta}''(\hat \rho_*,\hat\theta) a +\frac{1}{2}u_{\theta\theta}''(\hat\rho_*,\hat\theta) \right)(\Delta \theta)^2 +o(\Delta\theta)^2.
\]
This expression is maximized at $a=\tilde\rho_{FB}'(\theta)$ when $(\Delta \theta)^2\rightarrow 0$.
\end{proof}

\section{Quadratic Payoffs}\label{AppQuad}

\begin{lemma}\label{l:d2u}
Under Assumption \ref{quad}, 
\begin{enumerate}
\item $u_*(\theta)$ is continuously differentiable in $\theta$ for all $\theta \in [0,1]$;
\item $u_*(\theta)$  is twice continuously differentiable in $\theta$ for almost all $\theta \in [0,1]$, with
\begin{equation}\label{d2u}
u_*''(\theta)=(2\rho'_{FB}(\theta)-\rho'_*(\theta))\rho'_*(\theta),\text{ if }|\rho_{FB}(\theta)-\rho_R(\theta)| \neq \ell.
\end{equation}
\end{enumerate}
\end{lemma}
\begin{proof}
Since
\begin{equation*}\label{dstartheta}
\rho_*(\theta)=
\begin{cases}
\rho_R(\theta)+\ell, &\text{if } \rho_{FB}(\theta)-\rho_R(\theta) > \ell,\\
\rho_{FB}(\theta), &\text{if } | \rho_{FB}(\theta)-\rho_R(\theta)| \leq \ell,\\
\rho_R(\theta)-\ell, &\text{if } \rho_{FB}(\theta)-\rho_R(\theta) < -\ell,\\	
\end{cases}	
\end{equation*}
we have
\begin{align}
\rho'_*(\theta)=&
\begin{cases}
\rho'_R(\theta), &\text{if } |\rho_{FB}(\theta)-\rho_R(\theta)| > \ell,\\
\rho'_{FB}(\theta), &\text{if } | \rho_{FB}(\theta)-\rho_R(\theta)| < \ell.\\	
\end{cases}\label{dprime}
\end{align}
Further, $u_*(\theta)$ is continuously differentiable in $\theta$ for all $\theta \in [0,1]$, with 
\begin{equation*}\label{u*'}
u_*'(\theta)=
\begin{cases}
\rho'_{FB}(\theta) \rho_*(\theta) + (\rho_{FB}(\theta) - \rho_*(\theta))\rho'_*(\theta), &\text{if } |\rho_{FB}(\theta)-\rho_R(\theta)| \neq \ell,\\
\rho'_{FB}(\theta) \rho_*(\theta), &\text{if } | \rho_{FB}(\theta)-\rho_R(\theta)| = \ell.\\	
\end{cases}
\end{equation*}
Finally, since $\rho_*(\theta)$ is twice continuously differentiable in $\theta$ everywhere except at most two states where $|\rho_{FB}(\theta)-\rho_R(\theta)| = \ell$, it follows that $u_*(\theta)$ is twice continuously differentiable everywhere except at most these two states, with $u_*''(\theta)$ given by (\ref{d2u}).
\end{proof}

\begin{lemma}\label{elldelta}
If the receiver is highly responsive, $\ell$ is strictly increasing in $\delta$.
\end{lemma}
\begin{proof}
If $\delta=0$, then $u_*''(\theta)<0$ for all $\theta\in [0,1]$. By Corollary \ref{c:separation}, the expected joint payoff is strictly higher under complete pooling than under full separation. Thus, Lemma~\ref{mondelta} implies that $\overline{v}-\underline{v}_S-\underline{v}_R>0$ and  $\ell$ is strictly increasing in $\delta$ for all $\delta\in [0,1)$.	
\end{proof}

\begin{proof}[Proof of Corollary~\ref{c:cs}]
Since, by Lemma \ref{elldelta}, the relational leeway $\ell$ increases with $\delta$, the set of high-conflict states consists of up to one interval that shrinks and eventually vanishes as $\delta$ increases (Figure \ref{Fig:UpBias}). Specifically, there exist $\delta^A,\delta^B \in (0,1)$ such that $\delta^A>\delta^B$, and the set of high-conflict states is
\begin{equation*}
	X=\begin{cases}
		\emptyset,&\text{if } \delta\in(\delta^A,1),\\
		[0,\theta_{\mathrm{c}})\text{ for some $\theta_{\mathrm{c}}\in(0,1)$},&\text{if } \delta\in (\delta^B,\delta^A),\\
		[0,1],&\text{if }\delta\in[0,\delta^B).
	\end{cases}
\end{equation*}
By Proposition~\ref{pr:sb}, full separation is optimal if $\delta\in(\delta^A,1)$, and complete pooling is optimal if $\delta\in[0,\delta^B)$. Moreover, it is optimal to pool the states below $\theta_*\in (\theta_{\mathrm{c}},1]$ and separate the rest if $\delta \in (\delta^B,\delta^A)$. The expected joint payoff under this message rule is
\begin{equation*}
v_* =F(\theta_*) u_*(m_* )+\int_{\theta_*}^1u_*(\theta )dF(\theta).
\end{equation*}
Thus, taking into account that 
\begin{equation*}
	\frac{dm_*}{d\theta_*}=\frac{f(\theta_*)}{F(\theta_*)}(\theta_*-m_*),
\end{equation*}
we have
\begin{equation}\label{FOC1}
\begin{aligned}
	\frac{dv_*}{d\theta_*}&=F(\theta_*)u_*'(m_* )\frac{dm_*}{d\theta_*}+f(\theta_*)u_*(m_* )-f(\theta_*)u_*(\theta_* )\\
	&=f(\theta_*)\left(u_*'(m_* )(\theta_*-m_*)+u_*(m_* )-u_*(\theta_* )\right),
\end{aligned}	
\end{equation}
and
\begin{align*}
	\frac{d^2 v_*}{d\ell d\theta_*}&=f(\theta_*)\left(\frac{d u_*'(m_*)}{d \ell }(\theta_*-m_*)+\frac{d u_*(m_*)}{d \ell}\right)\\
	&=f(\theta_*)\left((\rho'_{FB}(m_*)-\rho'_*(m_*))(\theta_*-m_*)+(\rho_{FB}(m_*)-\rho_*(m_*))\right)\\
	&< f(\theta_*)\left((\rho'_{FB}(m_*)-\rho'_*(m_*))(\theta_{\mathrm{c}}-m_*)+(\rho_{FB}(m_*)-\rho_*(m_*))\right)=0,
\end{align*}
where the inequality holds because $\theta_*>\theta_{\mathrm{c}}>m_*$ and $\rho'_{FB}(m_*)<\rho'_*(m_*)$, and the last equality holds because $\rho_{FB}(\theta_{\mathrm{c}})=\rho_*(\theta_{\mathrm{c}})$ and $\rho_{FB}(\theta)-\rho_*(\theta)$ is linear in $\theta$ for $\theta\in (0,\theta_{\mathrm{c}})$.
So $\theta_*$ is non-increasing in $\ell$, and $d\theta_*/d\ell<0$ if $\theta_*<1$, as follows, for example, from Theorem 1 of \cite{edlin1998strict}. 
\end{proof}

\begin{corollary}\label{c:csbeta}
Suppose the receiver is highly responsive and $\theta_{\mathrm{c}}\in (0,1)$. 
Keeping $\theta_{\mathrm{c}}$ constant, $\theta_*$ is strictly increasing in $\alpha={\rho'_{R}(\theta)}/{\rho'_{FB}(\theta)}$ if $\theta_*<1$. Moreover, $\theta_* \rightarrow \theta_{\mathrm{c}}$ as $\alpha \rightarrow 2$.
\end{corollary}
\begin{proof}
Denote $\rho_{FB}(\theta_{\mathrm{c}})=d_{\mathrm{c}}$, $\rho'_{FB}(\theta)=a_{FB}$, and $\rho'_R(\theta)=a_R$. Notice that
\[\rho_{FB} (\theta)=d_{\mathrm{c}}+a_{FB}(\theta-\theta_{\mathrm{c}})\text{ for all $\theta$},\]
and
\[
\rho_*(\theta)=
\begin{cases}
d_{\mathrm{c}}+a_R(\theta-\theta_{\mathrm{c}}), &\text{if $\theta<\theta_{\mathrm{c}}$},\\
d_{\mathrm{c}}+a_{FB}(\theta-\theta_{\mathrm{c}}), &\text{if $\theta\geq \theta_{\mathrm{c}}$}.
\end{cases}
\]
Thus, we can rewrite \eqref{FOC1} as follows:
\begin{align*}
\frac{dv_*}{d\theta_*}=&f(\theta_*)\bigg(\left[a_{FB}\ \rho_*(m_*) + a_R(\rho_{FB}(m_*) - \rho_*(m_*))\right](\theta_*-m_*) \\
&+ \rho_{FB}(m_*)\rho_*(m_*)-\frac{\rho_*^2(m_*)}{2}-\frac{\rho_{FB}^2(\theta_*)}{2}\bigg)\\
=&f(\theta_*) \left(a_R(a_R-2a_{FB})(\theta_{\mathrm{c}}-m_*)\left(\theta_*-\frac{\theta_{\mathrm{c}}+m_*}{2}\right)-a_{FB}^2\frac{(\theta_*-\theta_{\mathrm{c}})^2}{2}\right).
\end{align*}
Hence
\begin{align*}
\frac{d}{d\alpha}\left(\frac{1}{a_{FB}^2}\frac{dv_*}{d\theta_*}\right)=f(\theta_*)(\alpha-1)(\theta_{\mathrm{c}}-m_*)(2\theta_*-{\theta_{\mathrm{c}}-m_*})>0.
\end{align*}
So $\theta_*$ is non-decreasing in $\alpha$, and $d\theta_*/d\alpha>0$ if $\theta_*<1$, as follows from Theorem 1 of \cite{edlin1998strict}. Finally, if $\alpha \rightarrow 2$, then $\left.dv_*/d\theta_*\right|_{\theta_*=\theta_{\mathrm{c}}}\rightarrow 0$, implying that $\theta_* \rightarrow \theta_{\mathrm{c}}$.
\end{proof}

\section{Punishment}\label{a:pun}

\begin{proof}[Proof of Proposition~\ref{prop:spp}]
By Proposition~\ref{pr:opt}, $\overline v = \mathbb E [u_*(\rho_*(\mu_*(\theta)),\theta)]$.

Because the receiver can guarantee the payoff $u_R(\rho_R([0,1]),[0,1])$ by rejecting all transfers and choosing $\rho_R([0,1])$ in all periods, we have
\[\underline v_R\geq u_R(\rho_R([0,1]),[0,1]).\]

A sender's worst equilibrium with zero first-period ex-ante transfers ($\tau=0$) clearly exists. Let $\mu(\theta)$, $t(m)$, $\rho(m)$, $T(m)$, $v_S(m)$, $d^P$, and $m^P$ be used in the first period of such an equilibrium.
Define $V_S(\theta)$ as the expected payoff of the sender  if the first-period state is~$\theta$,
\begin{gather*}
	V_S(\theta)=(1-\delta)u_{S}(\rho(\mu(\theta )),\theta )-p(\theta),\\
\text{where }p(\theta)=(1-\delta)[t(\mu(\theta))+T(\mu(\theta) )]-\delta v_S(\mu(\theta)).
\end{gather*}
Condition C2 (a) requires that
\begin{equation}\label{ICap}
V_S(\theta) \geq  (1-\delta)u_{S}(\rho(\mu(\hat{\theta})),\theta )-p(\hat{\theta})\text{ for all $\theta,\hat{\theta}\in [0,1]$}. %
\end{equation}
As explained in the proof of Proposition~\ref{pr:impl}, this inequality holds if and only if 
\begin{equation}\label{non-decreasing}
\rho(\mu(\theta))\text{ is non-decreasing in }\theta,
\end{equation}
\begin{equation}\label{env}
V_S(\theta)=V_S(0)+(1-\delta)\int_0^\theta \frac{\partial u_S}{\partial \theta} (\rho(\mu(\tilde\theta)),\tilde\theta) d\tilde\theta\text{ for all }\theta\in [0,1]. 	
\end{equation}
Condition C2 (b), the constraint of Condition C4 (a) and the sender's constraint of Condition C5, and the constraint of Condition C4 (b), respectively imply that
\begin{gather}
V_S(\theta)\geq (1-\delta) u_S(d^P,\theta)+\delta\underline{v}_S\text{ for all }\theta\in [0,1],\label{lowerbound}\\
w(\rho(m),m)\leq L(\overline{v})\text{ for all } m\in\mu([0,1]),\label{selfenf}\\
w(d^P,m^P)\leq L(\overline{v}).\label{punishment}
\end{gather}
Thus, $\underline v_S$ is greater or equal than the value of the following problem
\begin{equation}\label{eq:relpr}
\begin{gathered}
\min_{\mu,\rho,V_S,m^P,d^P} \mathbb E \left[V_S(\theta)\right]\\
\text{subject to (\ref{non-decreasing}) -- (\ref{punishment})}.
\end{gathered}
\end{equation}

\begin{claim}\label{cl:relax}
There exists an optimal solution to the problem (\ref{eq:relpr}) such that $m^P\in \mu([0,1])$, $d^P=\rho (m^P)$, and (\ref{lowerbound}) holds with equality for $\theta \in m^P$.
\end{claim}
\begin{proof}
Given $\rho$ and $\mu$ that satisfy (\ref{non-decreasing}) and (\ref{selfenf}), define the function
\begin{equation}\label{hma}
h(m)=u_S(\rho(m),\theta(m))-\int_0^{\theta(m)} \frac{\partial u_S}{\partial \theta}(\rho(\mu(\tilde\theta)),\tilde\theta) d\tilde\theta 
\end{equation}
where $\theta(m)\in m$. Define
\begin{equation*}
m_\star \in \underset{m\in \mu([0,1])}{\arg\min} h(m)\text{ and }\theta_\star\in  m_\star.
\end{equation*} 
Hereafter, we assume that the infimum of $h$ is attained. If the infimum is not attained by any $m_\star$, a limiting argument,  as in the proof of Proposition~\ref{pr:impl}, needs to be made. It is easy to see that $\mu$, $\rho$, $m^P=m_\star$, $\theta^P=\theta_\star$, $d^P=\rho(m^P)$, and

\begin{equation}\label{VSb}
V_S(\theta)=(1-\delta)\left(u_S(\rho(m^P),\theta^P)+\int_{\theta^P}^\theta \frac{\partial u_S}{\partial \theta}(\rho(\mu(\tilde\theta)),\tilde\theta) d\tilde\theta\right)+\delta\underline{v}_S
\end{equation}
constitute a feasible solution to the problem (\ref{eq:relpr}). In particular, (\ref{VSb}) clearly satisfies (\ref{env}), and (\ref{lowerbound}) holds because
\begin{align*}
V_S(\theta) &= (1-\delta)\left(u_S(\rho(\mu(\theta)),\theta)-(h(\mu(\theta))-h(m^P))\right) + \delta \underline v_S\\
&\geq (1-\delta) u_S(\rho(m^P),\theta) +\delta \underline{v}_S,
\end{align*}
where the equality follows from (\ref{hma}) and (\ref{VSb}), and the inequality follows from (\ref{ICap}) evaluated at $\hat \theta=\theta^P$, where (\ref{ICap}) holds because (\ref{non-decreasing}) and (\ref{env}) hold.

Suppose for contradiction that there does not exist an optimal solution to (\ref{eq:relpr}) with the stated properties. Thus, in an optimal solution, $d^P\notin \rho(\mu([0,1]))$ and
\begin{equation}\label{intpun}
u_S(\rho(\mu(\theta)),\theta)-(h(\mu(\theta))-h(m_\star)) > u_S(d^P,\theta)\text{ for all }\theta\in [0,1].
\end{equation}
There are two cases to consider: $d^P\in (\rho(\mu(0)),\rho(\mu(1)) \setminus \rho(\mu([0,1]))$ and $d^P<\rho(\mu(0))$ (the case  $d^P>\rho(\mu(1))$ is analogous).

Suppose that $d^P\in (\rho(\mu(0)),\rho(\mu(1))) \setminus \rho(\mu([0,1]))$. Then there exists $\hat\theta\in (0,1)$ such that $d^P\in (\rho(\mu(\hat\theta-)),\rho(\mu(\hat\theta+))$. By continuity of $u_S$ and $V_S$, we have
\begin{equation*}
u_S(\rho(\mu(\hat\theta-)),\hat\theta)-(h(\mu(\hat\theta-))-h(m_\star))
	=u_S(\rho(\mu(\hat\theta+)),\hat\theta)-(h(\mu(\hat\theta+))-h(m_\star)).
\end{equation*}
Since $u_S(d,\theta)$ is concave in $d$ by Assumption~\ref{CS} and $h(m)$ is minimized at $m_\star$, this equality is incompatible with (\ref{intpun}) evaluated at $\hat\theta-$, leading to a contradiction.

Suppose that $d^P<\rho(\mu(0))$. The optimal $V_S$ is such that (\ref{lowerbound}) holds with equality for some~$\theta$,
\begin{equation*}
\min_{\theta\in [0,1]}(V_S(\theta)-(1-\delta)u_S(d^P,\theta))=\delta\underline v_{S},
\end{equation*}
which can be rewritten using (\ref{env}) as
\begin{equation*}
	(1-\delta)\min _{\theta\in [0,1] }\int_0^\theta \left(\frac{\partial u_S}{\partial \theta} (\rho(\mu(\tilde\theta)),\tilde\theta)-\frac{\partial u_S}{\partial \theta} (d^P,\tilde\theta) \right)d\tilde\theta=(1-\delta)u_S(d^P,0)+\delta\underline v _S-V_S(0).
\end{equation*}
Since $\partial^2 u_S(d,\theta)/\partial d \partial \theta>0$ and $\rho(\mu(\theta))>d^P$, the minimum is achieved at $\theta=0$. Moreover, (\ref{intpun}) implies that $u_S(d^P,0)<u_S(\rho(\mu(0)),0)$. Therefore, $u_S(\rho_-(0),0)\leq  u_S(d^P,0)$ because $\rho_-(0)\leq d^P$ by (\ref{punishment}), $d^P <\rho(\mu(0)) $   by supposition, and $u_S$ is concave in $d$. So an optimal $d^P<\rho(\mu(0))$ must be given by $\rho_-(0)$ to minimize $V_S(0)$, and thus $V_S$. But then we can modify $\mu$ and $\rho$ only in that $\mu$ separates $\theta=0$ and $\rho(\mu(0))$ is replaced with $\rho_-(0)$. Under this modification, we can support the same $V_S$ given by (\ref{env}) with $V_S(0)=(1-\delta)u_S(\rho_-(0),0)+\delta \underline v_S$, leading to a contradiction.
\end{proof}
Claim~\ref{cl:relax}, together with (\ref{VSb}), implies that $\underline v_S$ is greater or equal than
\begin{equation}\label{extrarelaxed}
\begin{gathered}
\min_{ \mu,\rho,\theta^P} \left\{ u_S(\rho(\mu(\theta^P)),\theta^P) + \mathbb E \left[ \int_{\theta^P}^\theta\frac{\partial u_S}{\partial \theta}({\rho}(\mu(\tilde\theta)),\tilde\theta)d\tilde \theta \right]\right\}\\
\text{subject to }\rho(\mu(\theta))\text{ is non-decreasing in }\theta,\\
w(\rho(m),m)\leq L(\overline{v})\text{ for all } m\in\mu([0,1]).
\end{gathered}
\end{equation}

\begin{claim}
There exists an optimal solution to the problem (\ref{extrarelaxed}) that solves the problem (\ref{eq:worst}).
\end{claim}
\begin{proof}
Consider an optimal solution $(\mu,\rho,\theta^P)$ to (\ref{extrarelaxed}).
Without loss of generality, $$m^P=\mu(\theta^P)=\{\theta:{\rho}(\mu(\theta))={\rho}(\mu(\theta^P))\},$$ otherwise we can modify the message and decision rules such that all states in $\{\theta:{\rho}(\mu(\theta))={\rho}(\mu(\theta^P))\}$ are pooled, the same decision $\rho(\mu(\theta))$ is induced  for all $\theta$, the constraints of (\ref{extrarelaxed}) hold, and the value of (\ref{extrarelaxed}) remains the same.
Moreover, 
\begin{equation*}\label{51}
{\rho}(\mu(\theta))=
\begin{cases}
\rho_-(\mu(\theta)), &\text{ if } \mu(\theta) > m^P,\\
\rho_+(\mu(\theta)), &\text{ if }\mu(\theta)<m^P,
\end{cases}
\end{equation*}
otherwise we can decrease the value of (\ref{extrarelaxed}) without violating the constraints either by decreasing $\rho(\mu(\theta))$ for $\mu(\theta)>m^P$ or by increasing $\rho(\mu(\theta))$ for $\mu(\theta)<m^P$.

Suppose for contradiction that there does not exist an optimal solution to (\ref{extrarelaxed}) with $\rho (m^P)\in \{\rho_-(m^P),\rho_+(m^P)\}$. Consider an optimal solution such that no other optimal solution has a strictly larger $m^P$ in the set order. 
If $\overline\theta^P=\sup m^P<1$, then some states adjacent to $m^P$ from above, say $(\overline\theta^P,\overline\theta^P+\varepsilon)$, must be pooled, otherwise we can decrease the value of (\ref{extrarelaxed}) by pooling states $(\overline\theta^P,\overline\theta^P+\varepsilon)$ and $m^P$ and inducing the same decision $\rho (m^P)$. Similarly, if $\underline \theta^P =\inf m^P>0$, then some states adjacent to $m^P$ from below, say $(\underline\theta^P-\varepsilon,\underline\theta^P)$, must be pooled. Notice that the objective function in (\ref{extrarelaxed}) is concave in $\rho(m^P)$; so we can decrease the value of (\ref{extrarelaxed}) without violating the constraints by changing $\rho(m^P)$ to at least one of the four values $\rho(\mu(\overline\theta^P +))$, $\rho(\mu(\underline \theta^P -))$, $\rho_+(m^P)$, $\rho_-(m^P)$, leading to a contradiction.
\end{proof}

It remains to show that a single-period punishment strategy profile from Proposition~\ref{prop:spp} can be supported in an equilibrium using the ex-ante transfers $\underline{\tau} _0$, $\underline{\tau} _S$, $\underline{\tau} _R$ given by
\begin{gather*}
\underline\tau_0=\underline\tau_S=\mathbb{E}[u_{S}(\rho_*(\mu_*(\theta)),\theta)-\underline{t}_0(\mu_*(\theta))]-\underline v_S,\\
(1-\delta)[\underline \tau_R  +\mathbb{E}[u_{R}(\rho_*(\mu_*(\theta)),\theta)+\underline{ t}_0(\mu_*(\theta))]] +\delta (\overline v - \underline v_S) =\underline v_R .
\end{gather*}
Condition C6 holds because the continuation payoffs are $v_S(m)=\underline{v}_S$ and $v_R(m)=\overline v-\underline{v}_S$. Condition C5 holds because Condition C6 holds and $T(m)=0$. The sender's (receiver's) constraint of Condition C1 holds with equality for $\underline \tau _0=\underline \tau _S$ (for $\underline \tau _R$). The receiver's (sender's) constraint of Condition C1 holds for $\underline \tau _0=\underline \tau _S$ (for $\underline \tau _R$) because it can be simplified to $\overline v\geq \underline{v}_S +\underline{v}_R$. 
Condition C2 (a) holds because $\underline\rho_j(\underline\mu_j(\theta))$  is non-decreasing in $\theta$ and $\underline{t}_j$ satisfies (\ref{eq:h}). 
Condition C2 (b) holds because by deviating to a message-transfer pair $(\hat{m},\hat{t})$ that is not observed on the equilibrium path, the sender induces $\underline d_j^P=\underline\rho_j(\underline m_j^P)$, which he can induce more cheaply on the equilibrium path with message $\underline m_j^P$ and zero interim transfer $\underline{ t}_{j}( \underline m_j^P) =0$, as required by (\ref{eq:ts}). Condition C4 (a) holds because $w(\underline \rho _j (m),m)\leq L (\overline v)$ for all $m\in \underline \mu _j ([0,1])$. Condition C4 (b) holds  because Condition C4 (a) holds and $\underline d_j ^P = \underline \rho_j (\underline m_j^P)$. Condition C3 holds because Condition C4 holds and $\underline{ t}_j(m)$ is non-positive.
\end{proof}

\begin{proposition}\label{prop:worst}
Suppose Assumption \ref{CS} holds and the state is uniformly distributed. Suppose also that $\rho_S(1/2) > \rho_R(1/2)$. Denote $\rho_S(\theta)=a_S \theta + b_S$ and $\rho_R(\theta)=a_R \theta + b_R$. Then there exists an optimal equilibrium in single-period punishment strategies where the sender's penal decision rule is $\underline\rho_S(m)=\rho_R(m)-\ell$ for all $m$ and the sender's penal message rule $\underline\mu_S$ pools the states below $\underline{\theta}_S$ and separates the states above $\underline{\theta}_S$ where
\begin{equation*}\label{eqn:worstt}
\begin{gathered}
\underline{\theta}_S =
\begin{cases}
0, &\text{ if }\ell > \frac{3 a_R^2}{32 a_S}-b_S+b_R,\\
\frac{a_R+\sqrt{a_R^2-8a_S (b_S-b_R+\ell)}}{2a_S}, &\text{ if } \ell < \frac{3 a_R^2}{32 a_S}-b_S+b_R.
\end{cases}
\end{gathered}
\end{equation*}
\end{proposition}

\begin{proof} Define $\overline\theta_\star=\sup m^P$ and $\underline\theta_\star=\inf m^P$, where $m^P$ solves (\ref{eq:worst}).
We first show by contradiction in each case that if $\overline\theta_\star<1$, then $\underline{\rho}_S(m^P)=\rho_R(m^P)-\ell$ and all states $\theta>\overline\theta_\star$ are separated by $\underline{\mu}_S$. Suppose that $\underline{\rho}_S(m^P)\in(\rho_R(m^P)-\ell,\rho_R(\overline\theta_\star)-\ell]$. If some states above $\overline\theta_\star$ are pooled, say $(\theta_1,\theta_2)$, we can decrease the value of (\ref{eq:worst}) by separating these states, as follows from
\begin{equation*}
\int_{\theta_1}^{\theta}\left(\rho_R\left(\theta_1+\theta_2\right)/2-\ell\right)d\tilde \theta>  \int_{\theta_1}^{\theta}\left(\rho_R \left(\tilde \theta \right)-\ell\right)d\tilde \theta\text{ for }\theta\in (\theta_1,\theta_2).
\end{equation*}
If all states above $\overline\theta_\star$ are separated, we can decrease the value of (\ref{eq:worst}) by pooling states $[\overline\theta_\star,\overline\theta_\star+\varepsilon)$  together with the states in $m^P$, and inducing the same decision $\rho_R(m^P)$ for all these states, leading to a contradiction.
 Next suppose that $\underline{\rho}_S(m^P)>\rho_R(\overline\theta_\star)-\ell$. Then some states adjacent to $\overline\theta_\star$ from above, say $(\overline\theta_\star,\hat \theta)$,  must be pooled, such that $\underline{\rho}_S(m^P)<\rho_R\left(\left(\overline\theta_\star+\hat \theta\right)/2\right) -\ell$. But then we can decrease the value of (\ref{eq:worst}) by separating states $(\hat \theta-\varepsilon, \hat \theta)$, as follows from
\begin{small}\begin{gather*}
 \int_{\overline\theta_\star}^{\theta}\left(\rho_R \left(\overline\theta_\star+\hat \theta \right)/2-\ell\right)d\tilde \theta>  \int_{\overline\theta_\star}^{\theta}\left(\rho_R\left(\overline\theta_\star+\hat \theta -\varepsilon\right)/2-\ell\right)d\tilde \theta\text{ for }\theta\in (\overline\theta_\star,\hat \theta -\varepsilon),\\
\int_{\overline\theta_\star}^{\theta}\left(\rho_R\left(\overline\theta_\star+\hat \theta\right)/2-\ell\right)d\tilde \theta>  \int_{\overline\theta_\star}^{\hat \theta -\varepsilon}\left(\rho_R\left(\overline\theta_\star+\hat \theta -\varepsilon\right)/2-\ell\right)d\tilde \theta +  \int_{\hat \theta -\varepsilon}^{\theta}\left(\rho_R(\tilde \theta)-\ell\right)d\tilde \theta\text{ for }\theta\in (\hat \theta -\varepsilon,\hat \theta).
\end{gather*}\end{small}

Analogously, we can show that if $\underline\theta_\star>0$, then $\underline{\rho}_S(m^P)=\rho_R(m^P)+\ell$ and all states $\theta<\underline\theta_\star$ are separated by $\underline{\mu}_S$. This implies that either $\underline\theta_\star=0$ or $\overline\theta_\star=1$.

Thus, the sender's worst equilibrium payoff $\underline{v}_S$ is achieved either by a message rule that pools the states below $\underline{\theta}^L$ (and separates the rest) and decision rule $\underline\rho(m)=\rho_R(m)-\ell$, or by a message rule that pools the states above $\underline \theta^H$ (and separates the rest) and decision rule $\underline\rho (m)=\rho_R(m)+\ell$. Computation reveals that the value of (\ref{eq:worst}) unde pooling interval $[0,\underline \theta)$ and decision rule $\underline \rho (m) =\rho_R(m)-\ell$  is smaller than the value of (\ref{eq:worst}) under pooling interval $(1-\underline \theta ,1]$ and decision rule $\underline \rho (m) =\rho_R(m)+\ell$ for all $\underline \theta \in [0,1]$ if $a_S/2+b_S>a_R/2+b_R$. Moreover, the value of (\ref{eq:worst}) is minimized for $\underline{\theta}^L \in [0,1]$  at either $\underline{\theta}^L=0$ or 
\begin{equation*}
\underline{\theta}^L = \frac{a_R+\sqrt{a_R^2-8a_S (b_S-b_R+\ell)}}{2a_S}<1,
\end{equation*}
where the inequality follows from the assumption $a_S/2+b_S>a_R/2+b_R$. Further computation then produces $\underline \theta_S$, as defined in Proposition~\ref{prop:worst}.
\end{proof}

\section{Public Information}\label{appendix_t}
\begin{proof}[Proof of Proposition \ref{prop:L_info}] 
Suppose, for the sake of argument, that $L(\overline{v})$, as defined in Section \ref{sec:impl}, takes the same value under $\psi$ and $\hat{\psi}$. We will show that the best equilibrium joint payoff is higher and the worst monotone equilibrium payoffs are smaller under $\hat{\psi}$ than under~$\psi$. Specifically,  $\hat{\overline{v}}\geq\overline{v}$, $\hat{\underline{v}}_R\leq\underline{v}_R$, and $\hat{\underline{v}}_S\leq\underline{v}_S$.  This implies that $\hat{L}(\hat{\overline{v}})\geq L(\overline{v})$. The proposition follows easily from this observation.

The best equilibrium joint payoff $\overline{v}$ under $\psi$ can be supported by an equilibrium in single-period punishment strategies such that $\rho_*(\mu_*(\theta))$ is induced in each period on the equilibrium path, by application of Proposition~\ref{prop:spp} to each realization of signal $\psi$. Since $\rho_*(\mu_*(\theta))$ is non-decreasing in $\theta$ on $[0,1]$, it can be supported in an equilibrium under less informative signal $\hat{\psi}$ by application of an analogue of Proposition~\ref{pr:impl} to each realization of signal $\psi$; so, $\hat{\overline{v}}\geq \overline{v}$.

By Proposition~\ref{prop:spp}, the receiver's worst equilibrium payoffs under $\hat\psi$ and $\psi$ are 
\begin{equation*}
	\hat{\underline{v}}_R=\mathbb{E}[u_R(\rho_R(\hat{\psi}(\theta)),\theta)]\leq  \mathbb{E}[u_R(\rho_R(\psi(\theta)),\theta)]=\underline{v}_R,
\end{equation*}
where the inequality holds because $\psi$ is more informative than $\hat{\psi}$.

By a similar argument to the proof of Proposition~\ref{prop:spp}, the sender's worst equilibrium payoff under $\psi$ can be supported by $\tau=0$, $T\left(m\right)=0$, and $v_S\left(m\right)=\underline{v}_S$; that is, the sender may refuse to make any ex-ante or ex-post transfers, and the worst punishment for him would involve zero transfers from the receiver and the worst continuation payoff. Let $\mu \left(\theta\right)$ and $\rho\left(m\right)$ be penal message and decision rules that support this equilibrium. By assumption $\rho(\mu(\theta))$ is non-decreasing in $\theta$. 
Then the interim transfer $t(\mu(\theta))$ is defined by (\ref{eq:h}) and (\ref{eq:ts}) given that the set of states is $\psi(\theta)\subset [0,1]$ rather than $[0,1]$:
\begin{equation}\label{tphi}
\begin{gathered}
t(m)=h(m)-\min_{m\in \mu(\psi(\theta))} h(m),\\
h(m)=u_S(\rho(m),\theta(m))-\int_0^{\theta(m)} \frac{\partial u_S}{\partial \theta} (\rho(\mu(\tilde\theta)),\tilde\theta) d\tilde\theta, 
\end{gathered}
\end{equation}
where $\theta(m)\in m$.
The message and decision rules $\mu(\theta)$ and $\rho(m)$ such that $\rho(\mu(\theta))$ is non-decreasing in $\theta$ can be supported in equilibrium under $\hat{\psi}$ using the interim transfer rule $\hat{t}(m)$ that differs from $t(m)$ given by (\ref{tphi}) only in that the minimum of $h$ is taken over $m\in \mu(\hat{\psi}(\theta))$ rather than over $m\in\mu(\psi(\theta))$. 
Since  $\psi(\theta)\subset \hat{\psi}(\theta)$ for all $\theta\in [0,1]$ by the definition of more informative signals, we have $\hat{t}(\mu(\theta))\geq t(\mu(\theta))$ for all $\theta\in [0,1]$, and thus
\begin{equation*}
\hat{\underline{v}}_S\leq\mathbb{E}[u_{S}(\rho(\mu(\theta )),\theta )-\hat{t}(\mu(\theta))]\leq \mathbb{E}[u_{S}(\rho(\mu(\theta )),\theta )-{t}(\mu(\theta))] =\underline{v}_S. \qedhere
\end{equation*}
\end{proof}

\end{appendices}
\bibliographystyle{econ}
\bibliography{relational}	
\end{document}